\documentclass[10pt]{IEEEtran}

\usepackage[figure, linesnumbered]{algorithm2e}
\usepackage{amsfonts}
\usepackage{amsmath}
\usepackage{amsthm}
\usepackage{multirow}
\usepackage{cite,booktabs}
\usepackage{hyperref}
\usepackage{bussproofs}
\usepackage{graphicx,caption}
\usepackage{subfig}
\usepackage{xspace}
\usepackage[usenames, dvipsnames]{color}
\usepackage{tikz}
\usepackage{scalerel}
\usepackage{ marvosym }

\usetikzlibrary{shapes, trees, positioning}

\makeatletter
\renewcommand{\fnum@figure}{Figure \thefigure}
\makeatother

\newtheorem{lemma}{Lemma}
\newtheorem{example}{Example}
\newtheorem{definition}{Definition}

\newtheorem{theorem}{Theorem}
\newtheorem{proposition}{Proposition}

\newcommand{\qubos}{\textsf{Qubos}\xspace}
\newcommand{\quantor}{\textsf{Quantor}\xspace}
\newcommand{\qell}{\textsf{QELL}\xspace}
\newcommand{\hqspre}{\textsf{HQSpre}\xspace}

\newcommand{\bloqqer}{\textsf{Bloqqer}\xspace}

\newcommand{\ghostqcegar}{\textsf{GhostQ}\xspace}

\newcommand{\qsts}{\textsf{QSTS}\xspace}
\newcommand{\qstsdefnobreaksym}{\textsf{QSTS}\xspace}
\newcommand{\depqbfbat}{\textsf{DepQBF}\xspace}
\newcommand{\caqe}{\textsf{CAQE}\xspace}
\newcommand{\qesto}{\textsf{QESTO}\xspace}
\newcommand{\rareqs}{\textsf{RAReQS}\xspace}
\newcommand{\dynqbf}{\textsf{DynQBF}\xspace}
\newcommand{\heretic}{\textsf{Heretic}\xspace}
\newcommand{\ijtihad}{\textsf{Ijtihad}\xspace}
\newcommand{\quterandom}{\textsf{Qute}\xspace}
\newcommand{\revqfun}{\textsf{Rev-Qfun}\xspace}
\newcommand{\depqbfprefixopt}{\textsf{DepQBF}\xspace}

\title{Expansion-Based QBF Solving  Without Recursion\thanks{This work has been supported by the Austrian Science Fund (FWF)
    under projects W1255-N23, S11406-N23, S11408-N23, and S11409-N23. 
  The paper will appear in the \textbf{proceedings} of FMCAD 2018, IEEE.}
}

\author{\IEEEauthorblockN{Roderick Bloem, Nicolas Braud-Santoni, Vedad Hadzic,}
\IEEEauthorblockA{TU Graz\\
}
\and
\IEEEauthorblockN{Uwe Egly, Florian Lonsing,}
\IEEEauthorblockA{TU Wien\\}
\and
\IEEEauthorblockN{Martina Seidl,}
\IEEEauthorblockA{JKU Linz\\
}}

\tikzset{
  treenode/.style = {align=center, inner sep=0pt, text centered, font=\tiny},
  univ/.style = {treenode, circle, black, draw=black, fill=white, minimum width=0.3em, minimum height=0.3em},
  exis/.style = {treenode, square, black, draw=black, fill=white, minimum width=0.3em, minimum height=0.3em},
  leaf/.style = {treenode, diamond, black, draw=none, fill=none}
}

\newlength{\treedistv}
\newlength{\treedisth}
\newlength{\treescript}
\newcommand{\ts}[1]{\scaleto{#1}{\treescript}}

\begin{document}
\maketitle

\begin{abstract}
In recent years, expansion-based techniques have been shown to be
very powerful in theory 
and practice for solving quantified Boolean formulas (QBF), 
the extension of propositional formulas with 
existential and universal quantifiers over 
Boolean variables. Such approaches partially expand one 
type of variable (either existential or universal) 
and pass the obtained formula to a SAT solver for deciding the QBF. 
State-of-the-art expansion-based solvers process the given formula 
quantifier-block wise and recursively apply expansion  
until a solution is found. 

In this paper, 
we present a novel algorithm for expansion-based QBF solving 
that deals with the whole quantifier prefix at once. 
Hence recursive applications of the expansion principle are avoided. 
Experiments indicate that the performance of our simple approach 
is comparable with the state of the art of QBF solving, especially 
in combination with other solving techniques. 

\end{abstract}

\section{Introduction}
\label{sec:intro}

Efficient tools for deciding the satisfiability of 
Boolean formulas (SAT solvers) are the core technology in many
verification and synthesis 
approaches~\cite{DBLP:journals/pieee/VizelWM15}. However, 
verification and synthesis problems are often beyond 
the complexity class NP as captured by SAT, 
requiring more powerful formalisms like \emph{quantified Boolean formulas} (QBFs). QBFs extend 
propositional formulas by universal and existential quantifiers over 
Boolean variables~\cite{DBLP:series/faia/BuningB09}
resulting in a decision problem that  is PSPACE-complete. 
Applications from verification and synthesis~\cite{DBLP:conf/vmcai/BloemKS14,DBLP:conf/cav/ChengHR16,DBLP:conf/tacas/ChengLR17,DBLP:conf/nfm/GasconT14,DBLP:conf/tacas/FaymonvilleFRT17,DBLP:conf/sat/HeymanSMLA14}, 
realizability checking~\cite{DBLP:journals/corr/FinkbeinerT15}, 
bounded model checking~\cite{DBLP:conf/sat/DershowitzHK05,DBLP:conf/cade/Zhang14}, and 
planning~\cite{DBLP:journals/amai/EglyKLP17,DBLP:conf/aaai/Rintanen07}
motivate the quest for efficient QBF solvers.   

Unlike for SAT, where \emph{conflict-driven clause learning} (CDCL) is 
the single dominant solving approach for practical problems, 
two dominant approaches exist for QBF solving. 
On one hand, CDCL has been successfully extended 
to QCDCL that enables clause and cube 
learning~\cite{DBLP:series/faia/GiunchigliaMN09,DBLP:conf/tableaux/Letz02,DBLP:conf/cp/ZhangM02}. 
On the other hand, 
\emph{variable expansion} has become very popular.
In short, expansion-based solvers eliminate one kind of variables by 
assigning them truth values and solve the resulting 
propositional formula with a SAT solver. For QBFs with one quantifier 
alternation (2QBF), a natural approach is to use two SAT solvers: one 
that deals with the existentially quantified variables and another one that deals with 
the universally quantified variables. For generalising this SAT-based  
approach to QBFs with an arbitrary number of quantifier alternations, 
expansion is recursively applied per quantifier block, 
requiring multiple SAT solvers. 
As noted by Rabe and Tentrup~\cite{DBLP:conf/fmcad/RabeT15},
these CEGAR-based approaches show poor performance for formulas with many 
quantifier alternations in general.  

In this paper, we present a 
novel solving algorithm based on non-recursive expansion 
for QBFs with arbitrary quantifier prefixes using only two SAT solvers.
Our approach of non-recursive expansion is
theoretically (i.e., from a proof complexity perspective) equivalent
to approaches that apply recursive expansion since both non-recursive
and recursive expansion rely on the $\forall$Exp+Res proof 
system~\cite{DBLP:conf/mfcs/BeyersdorffCJ14}. 
However, the non-recursive expansion has practical 
implications such as a modified search strategy. 
That is, the use of recursive or non-recursive expansion  results 
in different search strategies for the proof. 
With respect to proof search, there is an analogy to,
e.g., implementations of resolution-based CDCL SAT solvers that employ
different search heuristics. 

In addition, we implemented a hybrid 
approach that combines clause learning with non-recursive expansion-based 
solving for exploiting the power of QCDCL. 
Our experiments indicate that this hybrid approach 
performs very well, 
especially on formulas with multiple quantifier alternations. 

This paper is structured as follows. After a review of related work in 
the next section, we introduce the necessary preliminaries in 
Section~\ref{sec:prelims}. After a short recapitulation 
of expansion in Section~\ref{sec:exp}, our novel non-recursive expansion-based 
algorithm is presented in Section~\ref{sec:alg}. Implementation  
details are discussed in Section~\ref{sec:impl} together with 
a short discussion of the hybrid approach. In Section~\ref{sec:eval}
we compare our approach to state-of-the-art solvers.

\section{Related Work}
\label{sec:rw}

Already the early QBF solvers \qubos~\cite{DBLP:conf/fmcad/AyariB02} and
\quantor~\cite{DBLP:conf/fmcad/AyariB02} incorporate selective quantifier
expansion for eliminating one kind of quantification to reduce the given QBF
to a propositional formula. The resulting propositional formula is then solved
by calling a SAT solver once.  \qubos and \quantor impressively demonstrated
the power of expanding universal variables but also showed its enormous memory
consumption.  As a pragmatic compromise, bounded universal expansion was
introduced for efficient
preprocessing~\cite{DBLP:conf/sat/BubeckB07,DBLP:conf/sat/GiunchigliaMN10,DBLP:journals/jair/HeuleJLSB15,DBLP:conf/tacas/WimmerRM017}.

The first approach which uses two alternate SAT solvers $A$ and $B$ 
for solving 2QBF, 
i.e., QBFs of the form $\forall U \exists E.\phi$, was
presented in~\cite{DBLP:conf/sat/RanjanTM04}. Solver $A$ is initialised with $\phi$, $B$ 
with the empty formula. Both propositional formulas are incrementally refined
with satisfying assignments found by the other solver.  
If $A$ finds its formula unsatisfiable, then the QBF is false. Otherwise, 
the negation of the universal part of the satisfying assignment is passed 
to solver $B$.  
If solver $B$ finds its formula unsatisfiable, then the QBF is true. 
Otherwise, the existential part of the satisfying assignment is 
passed to solver $A$. Janota and Marques-Silva
generalised the idea of alternating SAT 
solvers~\cite{DBLP:conf/sat/JanotaS11} such that one solver deals with the existentially quantified 
variables and one solver deals with the universally quantified variables
exclusively. Solver $A$ gets \emph{instantiations} of $\phi$
in which the universal variables are assigned, and solver $B$ gets instantiations
of $\lnot\phi$ in which the existential variables are assigned. The satisfying 
assignment found by one solver is used to obtain a new instantiation 
for the other. This loop is repeated until one solver returns 
unsatisfiable. 
This approach realises a natural application of the counter-example guided 
abstraction refinement (CEGAR) paradigm~\cite{DBLP:journals/jacm/ClarkeGJLV03}. 
A detailed survey on 2QBF solving is given in~\cite{DBLP:conf/sat/BalabanovJSMB16}.  

A significant advancement  of expansion-based solving for QBF with 
an arbitrary number of quantifier alternations was made with the 
solver \rareqs~\cite{Janota20161,DBLP:conf/sat/JanotaKMC12}, which recursively applies the 
previously discussed 2QBF approach~\cite{DBLP:conf/sat/JanotaS11} 
for each quantifier alternation. The approach turned out to be 
highly competitive.\footnote{\url{http://www.qbflib.org}} For formalising this solving approach the calculus $\forall$Exp+Res
was introduced~\cite{DBLP:conf/mfcs/BeyersdorffCJ14}, and proof-theoretical investigations revealed 
the orthogonal strength of $\forall$Exp+Res and Q-resolution~\cite{DBLP:journals/iandc/BuningKF95}, the 
QBF variant of the resolution calculus that forms the basis for 
QCDCL-based solvers. 
Research on the proof complexity of QBF has identified an
exponential separation between Q-resolution and the $\forall$Exp+Res
system. There are families of QBFs for which any Q-resolution
proof has exponential size, in contrast to $\forall$Exp+Res
proofs of polynomial size, and vice versa. Hence these two systems have orthogonal strength.

Recent work successfully combines machine learning with this 
CEGAR approach~\cite{DBLP:conf/aaai/Janota18}.
Motivated by the success of expansion-based QBF solving, several 
other approaches~\cite{DBLP:conf/ijcai/JanotaM15,DBLP:conf/fmcad/RabeT15,DBLP:conf/aaai/0001JT16,DBLP:conf/sat/TuHJ15,DBLP:conf/sat/Tentrup16,DBLP:conf/cav/Tentrup17} have been presented that are based on levelised 
SAT solving, i.e., one SAT solver is responsible for the variables of 
one quantifier block. 
In this paper, we also introduce a solving approach that is based upon 
propositional abstraction but considers the whole quantifier prefix 
at once.  

\section{Preliminaries}
\label{sec:prelims}

The QBFs considered in this paper  are in prenex normal 
form $\Pi.\phi$ where 
$\Pi$ is a quantifier prefix $Q_1x_1Q_2x_2\ldots Q_nx_n$ 
over the set of variables $X = \{x_1,\ldots ,x_n\}$ with 
$Q_i \in \{\forall, \exists\}$ and $x_i \neq x_j$ for $i \neq j$.  
The propositional formula  $\phi$
contains only variables from $X$.
Unless stated otherwise, we do not make any assumptions on the structure
of $\phi$. Sometimes $\Pi.\phi$ is in \emph{prenex conjunctive normal form}
 (PCNF), 
i.e., $\Pi$ is a prefix as introduced before and 
$\phi$ is a conjunction of clauses. A clause is a disjunction of literals, 
and a literal is a variable or the negation of a variable. 
The prefix imposes the order $<_\Pi$ on the elements of $X$ such 
that $x_i <_\Pi x_j$ if $i < j$.
By $U_\Pi$ ($E_\Pi$) we denote the set of universally (existentially) 
quantified variables of the prefix $\Pi$. If clear from the 
context we omit the subscript $\Pi$.  
We assume the standard semantics of QBF. 
A QBF consisting of only the syntactic truth constant $\bot$ ($\top$) is false (true). 
A QBF $\forall x\Pi.\phi$ is 
true if $\Pi.\phi[x \leftarrow \top]$ and $\Pi.\phi[x \leftarrow \bot]$ 
are both true, where 
$\phi[x \leftarrow t]$ is the substitution of $x$ by $t$ in $\phi$.  
A QBF $\exists x\Pi.\phi$ is true if $\Pi.\phi[x \leftarrow \top]$ or $\Pi.\phi[x \leftarrow \bot]$ is true. 

Given a set of variables $X$, we call a function $\sigma\colon X \to \{\top,\bot,\epsilon\}$ an assignment for $X$. If there is an 
$x\in X$ with $\sigma(x) = \epsilon$ then $\sigma$ is a \emph{partial} assignment,
otherwise $\sigma$ is a \emph{full} assignment of $X$.  
Informally, $\sigma(x) = \epsilon$ means that $\sigma$ does not assign 
a truth value to variable $x$. 
A restriction 
$\sigma|_Y\colon Y\to \{\top,\bot,\epsilon\}$ of assignment 
$\sigma\colon X \to \{\top,\bot,\epsilon\}$ to $Y \subseteq X$ is defined 
by $\sigma|_Y(x) = \sigma(x)$ if $x \in Y$, otherwise  $\sigma|_Y(x) = \epsilon$. By $\Sigma_X$ we denote the set of all full assignments 
$\sigma\colon X\to \{\top,\bot\}$.
Let $\phi$ be a propositional formula over  $X$. 
By $\sigma(\phi)$ we denote the application of 
assignment $\sigma\colon X \to \{\top,\bot,\epsilon\}$
on  $\phi$, i.e., 
$\sigma(\phi)$ is the formula obtained by replacing variables 
$x \in X$ by $\sigma(x)$ if $\sigma(x) \in \{\top,\bot\}$ and performing 
standard 
propositional simplifications. Let $\phi, \psi$ be propositional formulas 
over the set of variables $X$. If for every full assignment $\sigma \in \Sigma_X$,
$\sigma(\phi) = \sigma(\psi)$ then $\phi$ and $\psi$ are equivalent. 
Let $\tau\colon X\to \{\top,\bot,\epsilon\}$ and 
$\sigma\colon Y\to \{\top,\bot,\epsilon\}$ be assignments such that 
for every $x \in X \cap Y$, $\tau (x) = \sigma(x)$ if $\tau(x) \not= \epsilon$ 
and $\sigma(x) \not= \epsilon$. 
Then the composite assignment of $\sigma$ and $\tau$ is denoted by $\sigma\tau\colon X\cup Y \to \{\top,\bot,\epsilon\}$ and for every propositional 
formula $\phi$ over $X \cup Y$, it holds that 
$\sigma\tau(\phi) = \tau\sigma(\phi) = \sigma(\tau(\phi)) = \tau(\sigma(\phi))$. Furthermore, $\sigma\sigma = \sigma$ for any 
assignment $\sigma$. 

\begin{example}
Let $\sigma\colon X\to\{\top,\bot,\epsilon\}$ be an assignment over variables $\{a, b, x, y\}$ defined 
by $\sigma(a) = \top$, $\sigma(b) = \epsilon$, 
$\sigma(x) = \top$, and $\sigma(y) = \epsilon$. 
The restriction $\tau = \sigma|_Y$ of $\sigma$ to $Y = \{x, y\}$ is given by 
$\tau(a) = \epsilon$, $\tau(b) = \epsilon$, 
$\tau(x) = \top$, $\tau(y) = \epsilon$. 
For the propositional formula $\phi = (x \lor a \lor y) \land (\neg x \lor \neg a \lor y) \land (\neg y \lor b)$, the application of $\sigma$ and $\tau$ 
on $\phi$ gives us $\sigma(\phi) = y \land (\neg y \lor b)$ and 
$\tau(\phi) = (\neg a \lor y) \land (\neg y \lor b)$.
\end{example}

\section{Expansion}
\label{sec:exp}

In the following, we introduce the notation and terminology used for 
describing expansion-based QBF solving in general, and 
the algorithm introduced in the next section in particular. 
We first define the notion of \emph{instantiation}
that is inspired by 
the axiom rule of the calculus $\forall$Exp+Res~\cite{DBLP:journals/tcs/JanotaM15}.

\begin{definition}
Let $\Pi.\phi$ be a QBF with prefix $\Pi = Q_1x_1\ldots Q_nx_n$ 
over the set of variables $X = \{x_1,\ldots ,x_n\}$ 
and $\sigma\colon Y \to \{\top,\bot,\epsilon\}$ with $Y \subseteq X$ an 
assignment. If $Y \subset X$, we extend the domain of $\sigma$ to X by setting
$\sigma(x) = \epsilon$ if $x \not\in Y$. 
 The \emph{instantiation of $\phi$ by $\sigma$}, denoted by $\phi^\sigma$,  is obtained from $\phi$ as follows:
\begin{enumerate}
\item all variables 
$x \in X$ with $\sigma(x) \neq \epsilon$ are set to $\sigma(x)$;
\item all variables $x \in X$ with $\sigma(x) = \epsilon$ are replaced by 
$x^\omega$ where annotation $\omega$ is uniquely defined by 
the sequence
$\sigma(x_{k_1})\sigma(x_{k_2})\ldots\sigma(x_{k_m})$ such that the set formed from the variables $x_{k_i}$ contains
all variables of $X$ with $x_{k_i} <_\Pi x$.
Furthermore, $x_{k_i} <_\Pi x_{k_j}$ if $k_i < k_j$. 
\end{enumerate}
\end{definition}

If we instantiate a QBF $\Pi.\phi$ with 
the full assignment  
$\sigma\colon U_\Pi\to \{\top, \bot\}$ 
of the universal variables, we obtain a propositional formula that contains only (possibly 
annotated) variables from $E_\Pi$.  
The dual holds for the instantiation by a full assignment $\sigma\colon E_\Pi\to \{\top, \bot\}$.

\begin{example}\label{lab:ex}
Given the QBF $\forall a \exists x \forall b \exists y .\phi$ with $\phi =  (
(x \lor a \lor y) \land (\neg x \lor \neg a \lor y) \land (\neg y \lor b))$. 
Then $ U = \{a, b\}$ and $E = \{x, y\}$. Let $\sigma\colon U\to\{\top,\bot,\epsilon\}$ be defined by $\sigma(a) = \top$ and $\sigma(b) = \bot$. Then 
$\phi^\sigma = (\neg x^{\top} \lor y^{\top\bot}) \land \lnot y^{\top\bot}$. Further, let 
$\tau\colon E\to \{\top,\bot,\epsilon\}$ with $\tau(x) = \bot$ and $\tau(y) = \bot$. Then $\phi^\tau = a$. Note that $a$ is not annotated because it 
occurs in the first quantifier block.  
\end{example}

Sometimes we want to remove the annotations again from 
an assignment or an instantiated 
formula. Therefore, we introduce the following notation. 
Let $\phi^\sigma$ be an instantiation by assignment $\sigma\colon X\to\{\top,\bot,\epsilon\}$ and $X^\sigma$ the set of annotated variables.
If we have an assignment $\tau\colon X^\sigma\to \{\top,\bot, \epsilon\}$, then 
we define $\tau^{-\sigma}\colon X\to \{\top,\bot\}$ by 
$\tau^{-\sigma}(x) = \tau(x^\sigma)$ for $x^\sigma \in X^\sigma$. 
If we have an instantiated formula 
$\phi^\sigma$, the $(\phi^\sigma)^{-\sigma}$ is the formula obtained by 
replacing every annotated variable $x^\sigma \in X^\sigma$ by $x$. 
In general, $(\phi^\sigma)^{-\sigma} \neq \phi$. 

\begin{lemma}
\label{lem:eq}
Let $\Pi.\phi$ be a QBF with variables $X$ and 
$\sigma\colon X\to \{\top,\bot,\epsilon\}$ be a partial assignment. 
Then $(\phi^\sigma)^{-\sigma}$ and $\sigma(\phi)$ are 
equivalent. 
\end{lemma}

\begin{proof}
By induction over the formula structure. For the base case let $\phi = x$ 
with $x \in X$. If $\sigma(x) = \epsilon$,  
 $\sigma(\phi) = x$ and
$\phi^\sigma = x$. 
Then  $(\phi^\sigma)^{-\sigma} = x$.
Otherwise, $\phi^\sigma = \sigma(x)$. Obviously, $\sigma(\phi) = \sigma(x) = (\sigma(x))^{-\sigma} \in 
\{\top,\bot\}$. The induction step naturally follows from the semantics of 
the logical connectives.   
\end{proof}

\begin{example}
Reconsider the propositional formula $\phi$ and assignments $\sigma, \tau$ 
from above (Example~\ref{lab:ex}). Then $(\phi^\sigma)^{-\sigma} = 
((\neg x^{\top} \lor y^{\top\bot}) \land \lnot y^{\top\bot})^{-\sigma} = (\neg x \lor y) \land \lnot y$. 
Furthermore, 
$(\phi^\tau)^{-\tau} = (a)^{-\tau} = a$. 
\end{example}

Finally, we specify the semantics of a QBF in terms of universal and 
existential expansion on which expansion-based QBF solving is founded.

\begin{lemma}\label{lem:unsat}
Let $\Phi = \Pi.\phi$ be a QBF with universal variables $U$. 
There is a set of assignments $A \subseteq \Sigma_U$ with
$\bigwedge_{\alpha \in A}\phi^\alpha$ is  
unsatisfiable if and only if $\Phi$ is false. 
\end{lemma}

The lemma above has a dual version 
for true QBFs. This duality plays a prominent role in 
our novel solving algorithm.

\begin{lemma}\label{lem:sat}
Let $\Phi = \Pi.\phi$ be a QBF with existential variables $E$. 
There is a set of assignments $S \subseteq \Sigma_E$ with 
$\bigvee_{\sigma \in S}\phi^\sigma$ is  
valid if and only if $\Phi$ is true. 
\end{lemma}

\begin{algorithm}[t]
\SetKwInOut{Input}{input}\SetKwInOut{Output}{output}
\Input{QBF $\Pi.\phi$ with universal variables $U$ and existential variables $E$}
\Output{truth value of $\Pi.\phi$}

$A_0 := \{\alpha_0\}$, where $\alpha_0\colon U \to \{\top,\bot\}$ is an arbitrary assignment \label{alg:basic:init}

$S_0 := \emptyset$

$i := 1$
\BlankLine
\While{true}{ \label{alg:basic:refloop-begin}
\BlankLine
$(\textit{isUnsat}, \tau) := \text{SAT}(\bigwedge_{\alpha \in A_{i-1}}\phi^\alpha)$ \label{alg:basic:sat-a}

\lIf{isUnsat} { return false}\label{alg:retfa}

$S_{i} := S_{i-1} \cup \{ (\tau|_{E^\alpha})^{-\alpha} \mid 
      \alpha\in A_{i-1}\}$ \label{alg:basic:upd-s}
\BlankLine
\BlankLine
$(\textit{isUnsat}, \rho) := \text{SAT}(\bigwedge_{\sigma \in S_{i}}\neg\phi^\sigma)$ \label{alg:basic:sat-s}

\lIf{isUnsat} { return true}\label{alg:rettr}

$A_i := A_{i-1} \cup \{ (\rho|_{U^\sigma})^{-\sigma} \mid 
      \sigma\in S_{i}\}$ \label{alg:basic:upd-a}

\BlankLine
\BlankLine
$i$++
\BlankLine

} \label{alg:basic:refloop-end}
\caption{Non-Recursive Expansion-Based Algorithm}
\label{alg:basic}
\end{algorithm}

\section{A Non-Recursive Algorithm for Expansion-Based QBF Solving}
\label{sec:alg}
The pseudo-code in Figure~\ref{alg:basic} summarises the basic idea of our 
novel approach for solving the QBF $\Pi.\phi$ with universal variables $U$ 
and existential variables $E$. 

First, an arbitrary assignment $\alpha_0$ for the universal 
variables is selected in Line~\ref{alg:basic:init}. The instantiation 
$\phi^{\alpha_0}$ is handed over to a SAT solver. If $\phi^{\alpha_0}$ is 
unsatisfiable,  
then $\Pi.\phi$ is false and the algorithm returns.
Otherwise, $\tau\colon E^{\alpha_0}\to \{\top,\bot\}$ is a 
satisfying assignment of $\phi^{\alpha_0}$. 
Let $\sigma_1$  denote the assignment $\tau^{-{\alpha_0}}$. 
Then ${\alpha_0}\sigma_1$ 
is a satisfying assignment of $\phi$. 

Next, the 
propositional formula $\lnot \phi^{\sigma_1}$ is handed over 
to a SAT solver for checking the validity 
of $\phi^{\sigma_1}$. 
 If $\lnot \phi^{\sigma_1}$ is unsatisfiable, then 
$\Pi.\phi$ is true and the algorithm returns. If 
$\lnot \phi^{\sigma_1}$ is satisfiable, then 
$\rho\colon U^{\sigma_1}\to \{\top,\bot\}$ is a
satisfying assignment of $\lnot \phi^{\sigma_1}$.
Let $\alpha_1$ denote the assignment $\rho^{-\sigma_1}$. 
Then $\alpha_1\sigma_1$
is a satisfying assignment for $\lnot\phi$.
The following lemma shows that $\alpha_0$ and $\alpha_1$ are different.

\begin{lemma}\label{lem:contra}
Let $\Pi.\phi$ be a QBF with universal variables $U$ and existential 
variables $E$. Further, let $\alpha\colon U\to \{\top,\bot\}$ 
be an assignment
such that the instantiation $\phi^\alpha$ is satisfiable 
and has the satisfying assignment $\tau\colon E^\alpha \to \{\top,\bot\}$.
Let $\sigma\colon E \to \{\top,\bot\}$ with
$\sigma = \tau^{-\alpha}$. 
Then $\alpha$ falsifies $(\lnot\phi^\sigma)^{-\sigma}$. 
\end{lemma}

\begin{proof}
Since $\phi^\alpha$ is satisfied by $\tau$, $\phi$ is satisfied by 
the composite assignment $\alpha\tau^{-\alpha} = \alpha\sigma$, and therefore
 $\lnot\phi$ is falsified 
by $\alpha\sigma$. 
Then $\alpha$ falsifies $\sigma(\lnot\phi)$. 
According to Lemma~\ref{lem:eq} $\sigma(\lnot\phi)$ is equivalent to 
$(\lnot\phi^\sigma)^{-\sigma}$. Then $\alpha$ also falsifies  $(\lnot\phi^\sigma)^{-\sigma}$. 
\end{proof}

In the next round of the algorithm, the propositional formula 
$\phi^{\alpha_0} \land \phi^{\alpha_1}$ is handed over to a SAT solver. 
If this formula is unsatisfiable, $\Pi.\phi$ is false and 
the algorithm returns. 
Otherwise, it is satisfiable under some assignment \(\tau\colon E^{\alpha_0} \cup E^{\alpha_1} \to \{\top,\bot\}\), then at least one 
new 
assignment $\sigma_2\colon E\to \{\top,\bot\}$ 
with  $\sigma_2 \neq \sigma_1$
can be extracted from $\tau|_{E^{\alpha_i}}$ with $0 \leq i \leq 1$. 
This assignment is 
then used for obtaining a new propositional formula $\phi^{\sigma_1} \lor \phi^{\sigma_2}$. To show the validity of this formula,  its negation 
is passed to a SAT solver. If this formula is unsatisfiable, 
$\Pi.\phi$ is true and the algorithm returns. 
Otherwise, it is satisfiable under the 
assignment \(\rho\colon U^{\sigma_1} \cup U^{\sigma_2} \to \{\top,\bot\}\). 
A new assignment $\alpha_2\colon U \to \{\top, \bot\}$ with 
$\alpha_2 \neq \alpha_1 \neq \alpha_0$
is obtained 
 from $\rho|_{A^{\sigma_i}}$ with $1 \leq i \leq 2$.
This
 assignment is then used in 
the next round of the algorithm. In this way, the 
propositional formulas $\bigwedge_{\alpha\in\Sigma_U}\phi^\alpha$ and 
$\bigvee_{\sigma\in\Sigma_E}\phi^\sigma$ are generated. 
If $\bigwedge_{\alpha\in A}\phi^\alpha$ is unsatisfiable 
for some $A \subseteq \Sigma_U$, by Lemma~\ref{lem:unsat} $\Pi.\phi$ is false. Dually, if 
$\bigvee_{\sigma\in S}\phi^\sigma$ is valid for some $S \subseteq \Sigma_E$, by Lemma~\ref{lem:sat} $\Pi.\phi$ is true. 
The algorithm iteratively extends the sets $A$ and $S$ by adding 
parts of satisfying assignments of $\phi$ to $S$ and parts of 
falsifying assignments to $A$. In particular, $A$ is extended by 
assignments of the universal variables and $S$ is extended 
by assignments of the existential variables. 
The order in which assignments are considered depends on the  used SAT solver. 

\begin{example}
\label{ex:short}
We show how to solve the QBF 
$\forall a\exists x\forall b \exists y.\phi$ with $E=\{x,y\}$, $U=\{a,b\}$, and $\phi= 
((a \lor x \lor y) \land (\lnot a \lor \lnot x \lor y) \land (b \lor \lnot y))$ with the 
algorithm presented above. This formula can be solved in two iterations:

\underline{Init:} We start with some random assignment 
$\alpha_0\colon U \to \{\top,\bot\}$, for example with $\alpha_0(a) = \top$ and $\alpha_0(b) = \bot$.  
 
\underline{Iteration 1:} The formula $\phi^{\alpha_0} = (\lnot x^\top \lor y^{\top\bot}) 
\land \lnot y^{\top\bot}$ is passed to a SAT solver and found satisfiable  
under the assignment
$\tau \colon E^{\alpha_0}\to \{\top,\bot\}$ with $\tau(x^\top) = \bot$ and $\tau(y^{\top\bot}) = \bot$. By removing the variable annotations we get assignment 
$\sigma_1  = (\tau|_{E^{\alpha_0}})^{-\alpha_0}$, where 
$\sigma_1\colon E\to \{\top,\bot\}$ with
$\sigma_1(x) = \bot$ and $\sigma_1(y) = \bot$. Based on this assignment
we obtain $\phi^{\sigma_1} = a$. The formula $\lnot \phi^{\sigma_1}$
 is passed to a SAT solver. It is satisfiable and has the satisfying assignment 
$\rho \colon U^{\sigma_1}\to \{\top,\bot\}$ with $\rho(a) = \bot$ and $\rho(b^\bot) = \top$, which we then reduce to
$\alpha_1 = (\rho|_{U^{\sigma_1}})^{-\sigma_1}$ where $\alpha_1\colon U\to \{\top,\bot\}$ with
$\alpha_1(a) = \bot$ and $\alpha_1(b) = \top$. 

\underline{Iteration 2:} The formula $\phi^{\alpha_0} \land \phi^{\alpha_1}
= (\lnot x^\top \lor y^{\top\bot}) \land \lnot y^{\top\bot}
\land (x^\bot \lor y^{\bot\top})
$ is passed to a SAT solver in the second iteration. It is satisfiable 
and one satisfying assignment is $\tau\colon E^{\alpha_0} \cup E^{\alpha_1} 
\to \{\top,\bot\}$ with $\tau(x^\top) = \bot, \tau(x^\bot) = \top, 
\tau(y^{\top\bot}) = \bot, \tau(y^{\bot\top}) = \bot$. From $\tau$, 
we can extract the assignment $\sigma_2  = (\tau|_{E^{\alpha_1}})^{-\alpha_1}$ where $\sigma_2\colon E\to \{\top,\bot\}$
with $\sigma_2(x) = \top$ and $\sigma_2(y) = \bot$. Note 
that for any choice of $\tau$, $\sigma_2 \neq \sigma_1$.  
Next, we 
construct $\phi^{\sigma_1} \lor \phi^{\sigma_2} = a \lor \neg a$. 
This formula is a tautology, so its negation 
that is passed to a SAT solver
is unsatisfiable, hence $\Pi.\phi$ is true.  
\end{example}

\begin{figure}
\centering
\begin{tikzpicture}[square/.style={regular polygon,regular polygon sides=4},-,
level distance = 0.55cm, font=\small] 
\tikzstyle{level 1}=[sibling distance=15mm] 
\tikzstyle{level 2}=[sibling distance=15mm] 
\tikzstyle{level 3}=[sibling distance=8mm]
\tikzstyle{level 4}=[sibling distance=8mm]

\node (A0) [univ, label=right:$a$] {}
    child[missing] { node [exis,label=right:$x^{\ts\top}$] {}
        child[dotted] {node [solid,univ,label=right:$b$] {}
            child[solid] {node [exis,label=right:$y^{\ts{\top\top}}$] {}
                child[dotted] {node [leaf,label=below:$\top$] {}}}
            child[solid] {node [exis,label=right:$y^{\ts{\top\bot}}$] {}
                child[dotted] {node [leaf,label=below:$\top$] {}}}
        }
    }
    child[solid]{ node [exis,label=right:$x^{\ts\bot}$] {}
        child[dotted] {node [solid,univ,label=right:$b$] {}
            child[missing] {node [exis,label=right:$y^{\ts{\bot\top}}$] {}
                child[dotted] {node [leaf,label=below:$\top$] {}}}
            child[solid] {node [exis,label=right:$y^{\ts{\bot\bot}}$] {}
                child[dotted] {node [leaf,label=below:$\alpha_0\sigma_1$] {}}}
        }
    }
;
\node (A1) [univ, below=\treedistv of A0, label=right:$a$] {}
    child[missing] { node [exis,label=right:$x^{\ts\top}$] {}
        child[dotted] {node [solid,univ,label=right:$b$] {}
            child[solid] {node [exis,label=right:$y^{\ts{\top\top}}$] {}
                child[dotted] {node [leaf,label=below:$\top$] {}}}
            child[solid] {node [exis,label=right:$y^{\ts{\top\bot}}$] {}
                child[dotted] {node [leaf,label=below:$\top$] {}}}
        }
    }
    child[solid]{ node [exis,label=right:$x^{\ts\bot}$] {}
        child[dotted] {node [solid,univ,label=right:$b$] {}
            child[solid] {node [exis,label=right:$y^{\ts{\bot\top}}$] {}
                child[dotted] {node [leaf,label=below:$\alpha_1\sigma_2$] {}}}
            child[solid] {node [exis,label=right:$y^{\ts{\bot\bot}}$] {}
                child[dotted] {node [leaf,label=below:$\alpha_0\sigma_1$] {}}}
        }
    }
;
\node (S1) [univ, right=\treedisth of A0, label=right:$a$] {}
    child[dotted] { node [solid,exis,label=right:$x$] {}
        child[solid] {node [univ,label=right:$b^{\ts\top}$] {}
            child[dotted] {node [solid,exis,label=right:$y$] {}
                child[missing] {node [leaf,label=below:$\bot$] {}}
                child[solid] {node [leaf,label=below:$\alpha_1\sigma_1$] {}}}
        }
        child[missing] {node [univ,label=right:$b^{\ts\bot}$] {}
            child[dotted] {node [solid,exis,label=right:$y$] {}
                child[solid] {node [leaf,label=below:$\bot$] {}}
                child[solid] {node [leaf,label=below:$\bot$] {}}}
        }
    }
;

\node (A2) [univ, below=\treedistv of A1, label=right:$a$] {}
    child{ node [exis,label=right:$x^{\ts\top}$] {}
        child[dotted] {node [solid,univ,label=right:$b$] {}
            child[solid] {node [exis,label=right:$y^{\ts{\top\top}}$] {}
                child[dotted] {node [leaf,label=below:$\alpha_2\sigma_3$] {}}}
            child[missing] {node [exis,label=right:$y^{\ts{\top\bot}}$] {}
                child[dotted] {node [leaf,label=below:$\top$] {}}}
        }
    }
    child[solid]{ node [exis,label=right:$x^{\ts\bot}$] {}
        child[dotted] {node [solid,univ,label=right:$b$] {}
            child[solid] {node [exis,label=right:$y^{\ts{\bot\top}}$] {}
                child[dotted] {node [leaf,label=below:$\alpha_1\sigma_2$] {}}}
            child[solid] {node [exis,label=right:$y^{\ts{\bot\bot}}$] {}
                child[dotted] {node [leaf,label=below:$\alpha_0\sigma_1$] {}}}
        }
    }
;
\node (S2) [univ, right=\treedisth of A1, label=right:$a$] {}
    child[dotted] { node [solid,exis,label=right:$x$] {}
        child[solid] {node [univ,label=right:$b^{\ts\top}$] {}
            child[dotted] {node [solid,exis,label=right:$y$] {}
                child[solid] {node [leaf,label=below:$\alpha_2\sigma_2$] {}}
                child[solid] {node [leaf,label=below:$\alpha_2\sigma_1$] {}}}
        }
        child[missing] {node [univ,label=right:$b^{\ts\bot}$] {}
            child[dotted] {node [solid,exis,label=right:$y$] {}
                child[solid] {node [leaf,label=below:$\bot$] {}}
                child[solid] {node [leaf,label=below:$\bot$] {}}}
        }
    }
;
\node (A3) [univ, below=\treedistv of A2, label=right:$a$] {}
    child{ node [exis,label=right:$x^{\ts\top}$] {}
        child[dotted] {node [solid,univ,label=right:$b$] {}
            child[solid] {node [exis,label=right:$y^{\ts{\top\top}}$] {}
                child[dotted] {node [leaf,label=below:] {}}}
            child[solid] {node [exis,label=right:$y^{\ts{\top\bot}}$] {}
                child[dotted] {node [leaf,label=below:] {\footnotesize \Lightning}}}
        }
    }
    child[solid]{ node [exis,label=right:$x^{\ts\bot}$] {}
        child[dotted] {node [solid,univ,label=right:$b$] {}
            child[solid] {node [exis,label=right:$y^{\ts{\bot\top}}$] {}
                child[dotted] {node [leaf,label=below:] {}}}
            child[solid] {node [exis,label=right:$y^{\ts{\bot\bot}}$] {}
                child[dotted] {node [leaf,label=below:] {}}}
        }
    }
;
\node (S3) [univ, right=\treedisth of A2, label=right:$a$] {}
    child[dotted] { node [solid,exis,label=right:$x$] {}
        child[solid] {node [univ,label=right:$b^{\ts\top}$] {}
            child[dotted] {node [solid,exis,label=right:$y$] {}
                child[solid] {node [leaf,label=below:$\alpha_2\sigma_2$] {}}
                child[solid] {node [leaf,label=below:$\alpha_2\sigma_1$] {}}}
        }
        child[solid] {node [univ,label=right:$b^{\ts\bot}$] {}
            child[dotted] {node [solid,exis,label=right:$y$] {}
                child[missing] {node [leaf,label=below:$\bot$] {}}
                child[solid] {node [leaf,label=below:$\alpha_3\sigma_3$] {}}}
        }
    }
;

\node(Legend00) [univ, below left=\treedisth+3mm and 4mm of S3, label=right:{$v$}] {};
\node(Legend01) [leaf, below right=0.4cm and 0.4cm of Legend00] {};

\node(Legend10) [univ, below=3mm of Legend01, label=left:{$v$}] {};
\node(Legend11) [leaf, below left=0.4cm and 0.4cm of Legend10] {};

\node(Legend20) [univ, below right=3mm and 2mm of Legend11, label=left:{$v$}] {};
\node(Legend21) [leaf, below=0.4cm of Legend20] {};

\node(L00) [leaf, above right=0.75cm and \treedisth+18mm of A0] {};
\node(L01) [leaf, above left=0.75cm and 15mm of A0] {};

\node(L10) [leaf, above right=0.5cm and \treedisth+18mm of A1] {};
\node(L11) [leaf, above left=0.5cm and 15mm of A1] {};

\node(L20) [leaf, above right=0.5cm and \treedisth+18mm of A2] {};
\node(L21) [leaf, above left=0.5cm and 15mm of A2] {};

\node(L30) [leaf, above right=0.5cm and \treedisth+18mm of A3] {};
\node(L31) [leaf, above left=0.5cm and 15mm of A3] {};

\draw[densely dashed] (L00) -- (L01) node[pos=0.25, label=above:{\scriptsize 
Counter-Models of  $\phi$}]{};

\draw[densely dashed] (L00) -- (L01) node[pos=0.75, label=above:{\scriptsize Models of $\phi$}]{};

\draw[densely dashed] (L00) -- (L01) node[pos=0.5, label=below:{\scriptsize Iteration 1}]{};
\draw[densely dashed] (L10) -- (L11) node[pos=0.5, label=below:{\scriptsize Iteration 2}]{};
\draw[densely dashed] (L20) -- (L21) node[pos=0.5, label=below:{\scriptsize Iteration 3}]{};
\draw[densely dashed] (L30) -- (L31) node[pos=0.5, label=below:{\scriptsize Iteration 4}]{};

\node [above right= 5mm and 1.4cm of L11, label=below:{$\sigma_1(x)=\top, \sigma_1(y)=\bot$}] {};
\node [above right= 5mm and 1.4cm of L21, label=below:{$\sigma_2(x)=\top, \sigma_2(y)=\top$}] {};
\node [above right= 5mm and 1.4cm of L31, label=below:{$\sigma_3(x)=\bot, \sigma_3(y)=\bot$}] {};
\node [above left= 5mm and 1.4cm of L10, label=below:{$\alpha_1(a)=\bot, \alpha_1(b)=\top$}] {};
\node [above left= 5mm and 1.4cm of L20, label=below:{$\alpha_2(a)=\top, \alpha_2(b)=\top$}] {};
\node [above left= 5mm and 1.4cm of L30, label=below:{$\alpha_3(a)=\top, \alpha_3(b)=\bot$}] {};

\draw[solid] (Legend00) -- (Legend01) node[pos=0.5, label=right:{\scriptsize \ \ $v$ set to $\bot$}]{};
\draw[solid] (Legend10) -- (Legend11) node[pos=0.5, label=right:{\scriptsize \ \ $v$ set to $\top$}]{};
\draw[dotted] (Legend20) -- (Legend21) node[pos=0.5, label=right:{\scriptsize \ \ $v$ unassigned}]{};

\end{tikzpicture}
\caption{Expansion trees relating the assignments found during solving 
the QBF $\forall a\exists x \forall b \exists y.\phi$ in Example~\ref{ex:long},
with initial assignment $\alpha_0(a) = \bot, \alpha_0(b) = \bot$.
The assignments shown in the leaves of the trees satisfy (left trees) 
or falsify (right trees) $\phi$.}
\label{fig:exp-trees}
\end{figure}
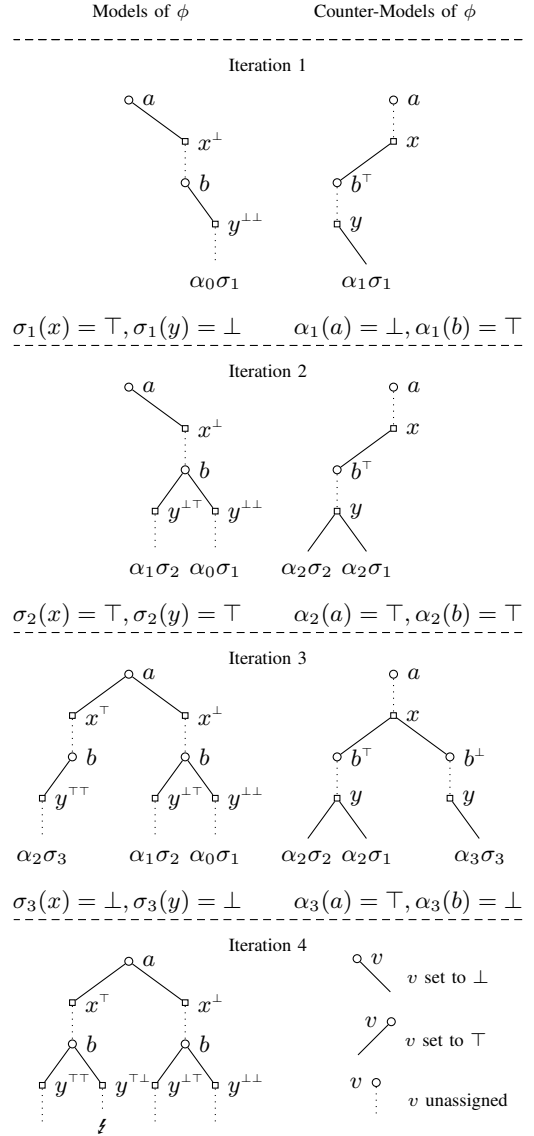

The soundness of our algorithm immediately follows from 
Lemmas~\ref{lem:unsat} and~\ref{lem:sat}: the algorithm returns false
(true) if, in some iteration $i$, it finds that the current partial 
expansion $\bigwedge_{\alpha \in A_{i-1}}\phi^\alpha$ 
(respectively $\bigwedge_{\sigma \in S_i}\lnot\phi^\sigma$)
is unsatisfiable. 

\begin{theorem}
The algorithm shown in Figure~\ref{alg:basic} is sound.
\end{theorem}

For showing that the algorithm also terminates, we argue that 
sets $A_i$ and $S_i$ increase in iteration $i+1$. To this end, 
we have to relate the variables of the QBF, the annotated variables 
as well as their assignments. Before we give the proof, we first 
consider another example in which we illustrate how the different 
assignments are related. 

\begin{example}
\label{ex:long}
We show one possible run of the algorithm presented above 
for the QBF $\Phi := \forall a\exists x\forall b \exists y.\phi$ with 
$$\phi := (a \land b \land \lnot x \land \lnot y) \lor (\lnot a \land x \land
(b \leftrightarrow y))$$
and how it iteratively generates the sets $\Sigma_U$ 
and $\Sigma_E$.
Figure~\ref{fig:exp-trees} shows the expansion trees that are implicitly 
built during the search. An expansion tree relates the variables 
of the partial expansion of $\Phi$ constructed from $A_i$ (left column)
and $S_i$ (right column). Solid edges indicate that the variable on the 
top has been set by an assignment from $A_i$ or $S_i$, and dotted edges indicate that the 
variable has to be assigned a value by the SAT solver. The order of the 
(annotated) variables in the expansion tree respects the order of the 
(original) variables in the prefix. 

\underline{Init}: For the initialisation of $A_0$,  an arbitrary assignment $\alpha_0\colon U \to \{\top,\bot\}$ is chosen. 
Let $\alpha_0(a) = \bot$ and $\alpha_0(b) = \bot$.

\underline{Iteration 1:} 
$\phi^{\alpha_0} := x^\bot \land \lnot y^{\bot\bot}$ is 
satisfiable. Assignment $\sigma_1\colon E \to \{\top,\bot\}$, with $\sigma_1(x) = \top$ and $\sigma_1(y) = \bot$, is extracted from model $\tau \colon E^{\alpha_1} \to \{\top,\bot\}$ and added to $S_1$.
Now $\phi^{\sigma_1} := \lnot a \land \lnot b^{\top}$ is checked 
for validity. Assignment $\alpha_1\colon U \to \{\top,
\bot\}$, with $\alpha_1(a) = \bot$ and $\alpha_1(b) = \top$,  
obtained from counter-example $\rho \colon U^{\sigma_1} \to \{
\top,\bot\} $ is added to $A_1$.

\underline{Iteration 2}: Next, $\phi^{\alpha_0} \land \phi^{\alpha_1}$ 
with $\phi^{\alpha_1} := x^\bot \land y^{\bot\top}$
is checked. From model $\tau \colon E^{\alpha_0} \cup E^{\alpha_1} \to \{\top,\bot\}$, again $\sigma_1$
can be extracted for $\phi^{\alpha_0}$. For $\phi^{\alpha_1}$ a new assignment 
$\sigma_2$ which is not in $S_1$ is found and added to $S_2$. In particular, we get $\sigma_2\colon E \to \{\top,\bot\}$ with $\sigma_2(x) = \top$ and $\sigma_2(y) = \top$. When the validity of $\phi^{\sigma_1} \lor \phi^{\sigma_2}$ 
with $ \phi^{\sigma_2} :=  \lnot a \land b^{\top}$
is checked, we get a counter-example $\rho \colon U^{\sigma_1} \cup U^{\sigma_2} \to \{\top,\bot\}$, from which $\alpha_2 \colon U \to \{\top,\bot\}$, with $\alpha_2(a) = \top$ and $\alpha_2(b) = \top$, can be extracted. Assignment $\alpha_2$ is added to $A_2$ leading to a new path in the left expansion tree (Iteration 3 in Figure~\ref{fig:exp-trees}).

\underline{Iteration 3:} Next, $\phi^{\alpha_0} \land \phi^{\alpha_1} \land \phi^{\alpha_2}$ 
with $ \phi^{\alpha_2} := \lnot x^\top \land \lnot y^{\top\top}$
is checked. 
From model $\tau \colon E^{\alpha_0} \cup E^{\alpha_1} \cup E^{\alpha_2} \to \{\top,\bot\}$, $\sigma_3 \colon E \to \{\top,\bot\}$ is extracted, satisfying $\phi^{\alpha_2}$. This assignment is different from both $\sigma_1$ and $\sigma_2$:
$\sigma_3(x) = \bot$ and $\sigma_3(y) = \bot$. This again results in a new branch of the expansion tree (see left expansion tree of Iteration~4 in Figure~\ref{fig:exp-trees}). The resulting formula $\phi^{\sigma_1} \lor \phi^{\sigma_2} \lor \phi^{\sigma_3}$ with $\phi^{\sigma_3} := a \land b^\bot$ is not valid, and from the counter-example $\rho \colon U^{\sigma_1} \cup U^{\sigma_2} \cup U^{\sigma_3} \to \{\top, \bot\}$ we get $\alpha_3 \colon U \to \{\top,\bot\}$ with $\alpha_3(a) = \top$ and $\alpha_3(b) = \bot$. 

\underline{Iteration 4:} Finally, the full 
expansion $\phi^{\alpha_0} \land \phi^{\alpha_1} \land \phi^{\alpha_2} \land \phi^{\alpha_3}$ with $\phi^{\alpha_3} := \bot$  is not satisfiable, 
meaning that the original formula 
$\forall a\exists x\forall b \exists y.\phi$ is false.
\end{example}

In the example above we saw that new assignments are generated in each iteration  because  $A_i$ and $S_i$ build models and counter-models 
of $\phi$. The following definition formalises the relationship between 
$A_i$ and $S_i$. 

\begin{definition}
Let $\Pi.\phi$ be a QBF over universally quantified variables $U$ 
and existentially quantified variables $E$. Further, let 
$A \subseteq \{ \alpha \mid \alpha \colon U \mapsto \{\top, \bot\}·\}$
and 
$S \subseteq \{ \sigma \mid \sigma \colon E \mapsto \{\top, \bot\}·\}$. 
If for every assignment $\sigma \in S$, 
there exists an assignment $\alpha \in A$ such that $\alpha\sigma(\lnot\phi)$ 
is true, then we say that $A$ \emph{completes} $S$. 
If for every assignment $\alpha \in A$, 
there exists an assignment $\sigma \in S$ such that $\alpha\sigma(\phi)$ 
is true, then we say that $S$ \emph{completes} $A$. 

\end{definition}

We now show that $S_{i}$ completes $A_{i-1}$ and $A_{i}$ completes $S_{i}$ if the algorithm does not terminate in iteration $i$ because of the 
unsatisfiability of the respective expansion.

\begin{lemma}\label{lem:compl}
Let $\Pi.\phi$ be a QBF over universally quantified variables $U$ 
and existentially quantified variables $E$. Further, let $A_{i-1}$ 
and $A_i$ with $A_{i-1} \subseteq A_i$ be 
two sets of full assignments to the universal variables and let $S_{i}$ be 
a set of full assignments to the existential variables 
obtained by  iteration $i$ during an execution of the algorithm 
shown 
in Figure~\ref{alg:basic}.

(1) If $\bigwedge_{\alpha \in A_{i-1}}\phi^\alpha$ is satisfiable, then
$S_{i}$ completes $A_{i-1}$, i.e., 
for every $\mu \in A_{i-1}$, there is an 
assignment $\nu \in S_i$ such that 
$\mu\nu(\phi)$ is true.

(2) If
$\bigwedge_{\sigma \in S_i}\lnot\phi^\sigma$ is satisfiable,
then $A_i$ completes $S_i$, i.e.,  
 for every $\nu \in S_i$, there is an 
assignment
$\mu \in A_{i}$ such that $\nu\mu(\lnot\phi)$ is true.
\end{lemma}
\begin{proof}
By contradiction. For (1), assume there is an assignment $\mu \in A_{i-1}$ 
such that  there is 
no assignment $\nu \in S_i$ with $\mu\nu(\phi)$ is true. 
By assumption  $\bigwedge_{\alpha \in A_{i-1}}\phi^\alpha$ is satisfiable, 
so there is a  satisfying assignment $\tau$ with 
$\tau |_{E^{\mu}}(\phi^{\mu})$ is true. 
Then also $\mu(\tau |_{E^{\mu}})^{-\mu}(\phi)$ 
is true. But 
$(\tau |_{E^{\mu}})^{-\mu} \in S_i$. 
For (2), assume that there is an assignment $\mu \in S_i$ 
such that there is no $\nu \in A_i$ with $\mu\nu(\lnot\phi)$ 
is true. The rest of the argument is similar as in (1). 
\end{proof}

Next, we show that the addition of new assignments $A'$ to 
a set $A$ of universal assignments forces a set $S$ of existential 
assignments to increase if some completion criteria hold. 

\begin{lemma}\label{lem:sinc}
Let $\Phi = \Pi.\phi$ be a QBF over universally quantified variables $U$ 
and existentially quantified variables $E$. Further, let $A \cup A'$ 
be a set of universal assignments such that $A \cap A' = \emptyset$
and $A' \not= \emptyset$. 
Let $S$ be a set of existential assignments
and assume that 
$\bigwedge_{\sigma \in S}\lnot\phi^\sigma$ has the satisfying assignment $\rho$,
$A' \subseteq \{(\rho |_{U^\sigma})^{-\sigma} \mid \sigma \in S\}$.

If $S$ completes $A$, and $A \cup A'$ completes $S$, and $\bigwedge_{\alpha \in A\cup A'}\phi^\alpha$ evaluates to true under assignment $\tau$, then there 
exists an assignment $\nu \in \{(\tau |_{E^\alpha})^{-\alpha} \mid \alpha \in A \cup A'\}$ with $\nu \not\in S$.  
\end{lemma}

\begin{proof}
By induction over the number of variables in $\Pi$. 

\emph{Base Case.} Assume that $\Phi$ has only one variable, i.e., 
$\Pi = Qx$. 
Note that $|A'| = 1$ because $x$ is outermost in the prefix and 
$A'$ is obtained from sub-assignments of $\rho$. 
If $Q = \forall$, then the elements of $A$ are full assignments of $\phi$, 
and $S$ is either empty, or it contains the empty assignment 
$\omega \colon \emptyset \mapsto \{\top, \bot\}$. 
Let 
$A' = \{\mu\}$.
If $S$ is empty, so is $A$
(because $S$ has to complete $A$). 
If $\tau$ is a satisfying assignment of $\phi^\mu$, then $\nu = \tau = \omega$ 
is the empty assignment and $\nu \not\in S$. 
Otherwise, $\omega \in S$. If there is an assignment $\alpha \in A$, 
then $\phi^\alpha \wedge \phi^\mu$ is a full expansion of $\Phi$. 
If this full expansion is true, then $\lnot \phi$ is unsatisfiable. 
Otherwise, $\phi^\alpha \wedge \phi^\mu$ is unsatisfiable. 
In both 
cases, the necessary preconditions for the lemma are not fulfilled. 
If $A = \emptyset$, then $\mu\omega(\lnot \phi)$ is true. Then 
$\phi^\mu$ is unsatisfiable, again violating a precondition. 
If $Q = \exists$, then $\mu = \omega$ and $A = \emptyset$.   
If $S = \emptyset$ and $\phi^\omega = \phi$ has the satisfying assignment 
$\tau$, then $\nu = \tau$ and $\nu \not\in S$. Otherwise, 
if there is an assignment 
$\sigma \in S$, then $\omega\sigma(\lnot \phi)$ is true, because 
$A \cup \{\mu\} = \{\omega\}$ completes $S$.
Hence, if assignment $\tau$ satisfies $\phi^\mu$, then $\nu = \tau$, so $\nu \not\in S$.  

\emph{Induction Step.} Assume the lemma holds for QBFs with $n$ variables. 
We show that it also holds for QBFs with $n+1$ variables. Let 
$\Phi = Qx\Pi.\phi$ be a QBF over existential variables $E$ and 
universal variables $U$ with $\Pi = Q_1x_1\ldots Q_nx_n$
and $A \cup A'$ and $S$ be as required ($S$ completes $A$, 
$A \cup A'$ completes $S$, $\bigwedge_{\alpha \in A \cup A'}\phi^\alpha$ 
has a satisfying assignment $\tau$, and 
$\bigwedge_{\sigma \in S}\lnot\phi^\sigma$ 
has a satisfying assignment $\rho$ from which $A'$ is obtained).

If $Q = \forall$, then all assignments $\alpha \in A'$ assign
the same value $t$ to $x$, i.e., $\alpha(x) = t$, because these 
assignments are extracted from assignment $\rho$ and since $x$ is the 
outermost variable of the prefix of $\Phi$, $\rho(x) = t$.
Further, let $A^t = \{ \alpha \in A \mid \alpha(x) = t\}$. It is easy to argue
that for $\Pi.\phi[x \leftarrow t]$ together with the assignment 
sets $A^t \cup A'$ and $S$ the induction hypothesis applies, i.e., there 
is an assignment $\nu \not\in S$ with 
$\nu \in \{(\tau' |_{E^\alpha})^{-\alpha} \mid \alpha \in A^t \cup A'\}$
where $\tau'$ is the part of $\tau$ that satisfies 
$\bigwedge_{\alpha \in A^t \cup A'}(\phi[x \leftarrow t])^\alpha$. 
Obviously, $\nu \in \{(\tau |_{E^\alpha})^{-\alpha} \mid \alpha \in A \cup A'\}$.

If $Q = \exists$, assume that $\tau(x) = t$. 
Let $\{ \sigma \in S \mid \sigma(x) = t\} \subseteq S^t \subseteq S$, 
and let $A^t \subseteq A$
such that the induction hypothesis applies to 
$\Pi.\phi[x \leftarrow t]$, $A^t \cup A'$, and $S^t$. Let
$\tau^t$ be those sub-assignments of $\tau$ 
that satisfy $\bigwedge_{\alpha \in A^t}\phi^\alpha$. Then there is
an assignment $\nu$ that can be extracted from $\tau^t$ with 
$\nu \not\in S^t$. Since $\nu(x) = t$, $\nu \not\in S$.
This concludes the proof. 
\end{proof}

This property also holds in the other direction, i.e., adding a set $S'$ of new assignments to $S$ will force the set $A$ to increase. 

\begin{lemma}\label{lem:ainc}
Let $\Phi = \Pi.\phi$ be a QBF over universally quantified variables $U$ 
and existentially quantified variables $E$. Further, let $S \cup S'$ 
be a set of existential assignments such that $S \cap S' = \emptyset$,
$S' \not= \emptyset$,
let $A$ be a set of universal assignments,
$\bigwedge_{\alpha \in A}\phi^\alpha$ has the satisfying assignment $\tau$,
$S' \subseteq \{(\tau |_{E^\alpha})^{-\alpha} \mid \alpha \in A\}$.

If $A$ completes $S$ and $S \cup S'$ completes $A$ and $\bigwedge_{\sigma \in S\cup S'}\lnot\phi^\sigma$ evaluates to true under assignment $\rho$, then there 
exists an assignment $\nu \in \{(\rho |_{U^\sigma})^{-\sigma} \mid \sigma \in S \cup S'\}$ with $\nu \not\in A$.  
\end{lemma}

\begin{proof}
The proof is analogous to the proof of Lemma~\ref{lem:sinc}.
\end{proof}

Now that we have identified the relations between the sets of universal and existential assignments, we use them to show that the algorithm from Figure~\ref{alg:basic} terminates.

\begin{theorem}
The algorithm shown in Figure~\ref{alg:basic} terminates for any
QBF $\Phi = \Pi.\phi$.
\label{thm:termination}
\end{theorem}
\begin{proof}
By induction over the number of iterations i, we argue that 
sets $A_{i-1} \subset A_i$ and $S_{i-1} \subset S_i$.

\emph{Base Case.} Let $i = 1$ and $A_0 = \{\alpha_0\}$. 
$S_0 \subset S_1$, because $S_0 = \emptyset$
and $\sigma_1 \in S_1$ is a satisfying assignment of $\phi^{\alpha_0}$
(if $\phi^{\alpha_0}$ is unsatisfiable, the algorithm terminates). 
$A_0 \subset A_1$ directly follows from Lemma~\ref{lem:contra}. 

\emph{Induction Step.} For $i > 1$, we argue that $S_i \subset S_{i+1}$. 
By induction hypothesis the theorem holds for iteration $i$, 
i.e., $A_i = A_{i-1} \cup A'$ with $A_{i-1} \cap A' = 
\emptyset$ and $A' \not= \emptyset$
and $S_i = S_{i-1} \cup S'$ with $S_{i-1} \cap S' = \emptyset$ and $S' \not= \emptyset$. 
Because of Lemma~\ref{lem:compl}, 
$S_{i}$ completes $A_{i-1}$, and $A_{i}$ completes $S_i$. Furthermore, 
if $\bigwedge_{\sigma \in S_i}\lnot\phi^\sigma$ is satisfiable under 
some assignment $\rho$ (otherwise the algorithm would terminate), 
by construction $A' \subseteq \{(\rho |_{U^\sigma})^{-\sigma} \mid \sigma \in S_i\}$. Hence, Lemma~\ref{lem:sinc} applies and if 
$\bigwedge_{\alpha \in A_i}\phi^\alpha$ is satisfiable under some assignment 
$\tau$ (otherwise the algorithm would immediately terminate), 
then there is an assignment $\nu \in \{(\tau |_{E^\alpha})^{-\alpha} \mid \alpha \in A_i\}$  with $\nu \not\in S_i$. 

The argument for $A_i \subset A_{i+1}$ is similar and uses the property shown in Lemma~\ref{lem:ainc}. 
\end{proof}

Note that the algorithm presented above does not make any assumptions 
on the formula structure, i.e., for a QBF $\Pi.\phi$ it is not 
required that $\phi$ is in conjunctive normal form. 
Without any modification, our algorithm also works on formulas in 
PCNF---as SAT solvers 
typically process formulas 
in CNF only, we focus on this representation for the rest of the paper. 

We conclude this section by arguing that the $\forall$Exp+Res~\cite{DBLP:conf/sat/JanotaM13,DBLP:journals/tcs/JanotaM15,DBLP:conf/mfcs/BeyersdorffCJ14} calculus 
yields the theoretical foundation of our algorithm for refuting 
a formula $\Pi.\phi$ in PCNF with universal variables $U$. 
The $\forall$Exp+Res calculus consists of two rules, the axiom rule
 \begin{prooftree}
    \AxiomC{}
    \UnaryInfC{$C^\alpha$}
\end{prooftree}
{
where $C$ is a clause of  $\phi$   and $\alpha\colon U \to \{\top,\bot\}$ is a 
universal assignment }
as well as  the resolution rule 
 \begin{prooftree}
    \AxiomC{$C_1 \lor x^\omega$}
    \AxiomC{$C_2 \lor \lnot x^\omega$}
    \BinaryInfC{$C_1 \lor C_2$}
 \end{prooftree}
A derivation in $\forall$Exp+Res is a sequence of clauses where each 
clause is either obtained by the axiom or  derived from 
previous clauses by the application of the resolution rule. 
A refutation of a PCNF $\Pi.\phi$ is a derivation of the empty clause. 

The application of the axiom instantiates the universal variables of 
one clause of $\phi$. If enough of these instantiations can be 
found in order to derive the empty clause by the application of the 
resolution rule, the QBF $\Pi.\phi$ is false. Our algorithm in 
Figure~\ref{alg:basic} does not instantiate individual clauses, 
but all clauses of $\phi$ at once with a particular assignment of 
the universal variables. Hence, when the SAT solver finds
$\psi_\forall = \bigwedge_{\alpha \in A_i}\phi^\alpha$ unsatisfiable for 
some $A_i$, not necessarily all clauses of $\psi_\forall$ are required to 
derive the empty clause via resolution, but only the minimal unsatisfiable 
core of $\psi_\forall$, i.e., a subset of the clauses such that the removal 
of any clause would make this formula satisfiable. 

\begin{proposition}
Let $\Pi.\phi$ be a false QBF. Further, let $\psi_\forall = \bigwedge_{\alpha \in A_i}\phi^\alpha$ be obtained by the application of the algorithm in Figure~\ref{alg:basic}.
Further, let $\psi_\forall'$ be an unsatisfiable core of $\psi_\forall$. Then there is 
a $\forall$Exp+Res refutation such that all clauses that are introduced by the 
axiom rule occur in $\psi_\forall'$.
\end{proposition}

\section{Implementation}
\label{sec:impl}

The algorithm described in Section \ref{sec:alg} 
is realised in the solver \ijtihad\footnote{The name \ijtihad refers to the effort of solving cases in Islamic law (for details see \url{https://en.wikipedia.org/wiki/Ijtihad}).}
The most recent version of \ijtihad is available at 
\begin{center}
\url{https://extgit.iaik.tugraz.at/scos/ijtihad} 
\end{center}
\noindent
The solver is implemented in C++ and currently processes formulas in PCNF
available in the QDIMACS format. 
For accessing SAT solvers, \ijtihad uses the 
IPASIR interface~\cite{DBLP:journals/ai/BalyoBIS16}, 
which makes changing the SAT solver very easy. 
The SAT solver used in all of our experiments is 
Glucose~\cite{DBLP:conf/ijcai/AudemardS09}. 
Although the base implementation does reasonably well, we have 
realised various optimizations to make \ijtihad even more viable
in practice. 
Some of them are discussed in the following. 
\begin{figure}[htp]
\centering
\begin{subfloat}
  \centering
  \includegraphics[clip,width=.6\columnwidth]{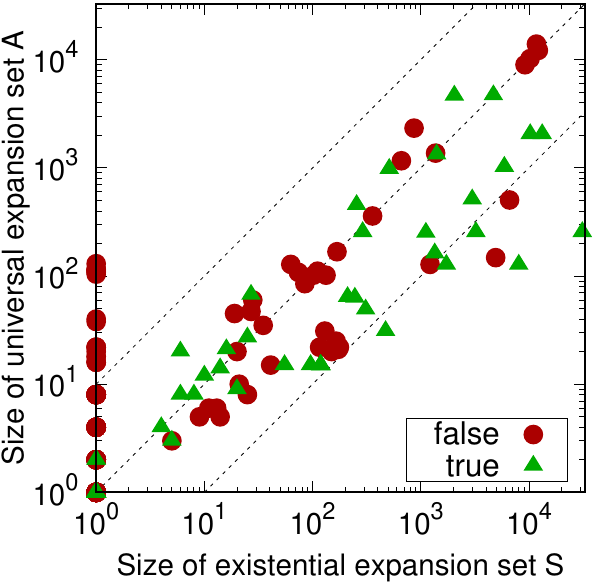}
  \caption{Sizes of sets \(S\) and \(A\)}
  \label{fig:scatterplot:1800:size}
\end{subfloat}%
\vspace{0.5cm}
\begin{subfloat}
  \centering
  \includegraphics[clip,width=.6\columnwidth]{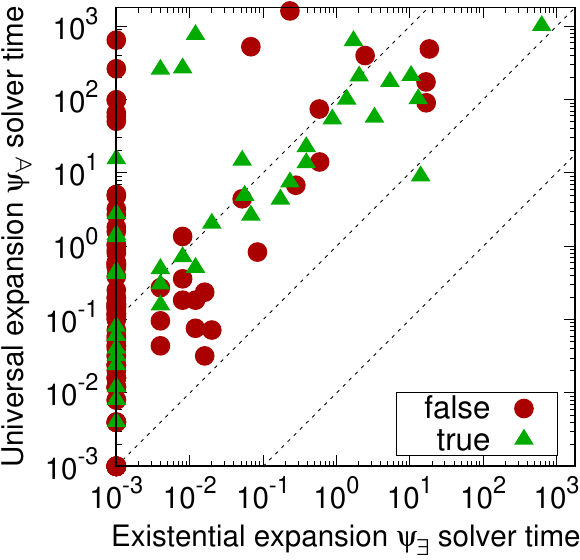}
  \caption{Time used for solving \(\psi_\exists\) and \(\psi_\forall\)}
  \label{fig:scatterplot:1800:time}
\end{subfloat}
\caption{Set sizes and time consumed during SAT calls for solved instances from QBFEVAL'17 preprocessed by \bloqqer}
\label{fig:scatterplot:1800:all}
\end{figure}

For solving a QBF $\Pi.\phi$, the basic algorithm shown in 
Figure~\ref{alg:basic}  adds instantiations of 
$\phi$ to  \(\psi_\forall = \bigwedge_{\alpha \in A_{i-1}}\phi^\alpha\) and \(\psi_\exists = \bigwedge_{\sigma \in S_{i}}\lnot \phi^\sigma\) in each 
iteration $i$ until the formula is decided. 
The calls to the SAT solver in 
Line \ref{alg:basic:sat-a} and Line \ref{alg:basic:sat-s} are done incrementally, 
i.e., we create two instances of the SAT solver and provide them with
the clauses stemming from  new instantiations of 
$\phi$ at each iteration. 
For simplicity, we omit indices of sets $A$ and $S$ and 
refer to an arbitrary iteration of the execution of the algorithm in the 
following discussion. 

Figure~\ref{fig:scatterplot:1800:all} relates set sizes of $A$ and $S$ 
as well as 
the accumulated time that one SAT solver needs to solve  \(\psi_\forall\)
with the time the other SAT solver needs to solve  \(\psi_\exists\)
for the formulas of the PCNF track of QBFEVAL'17 
(preprocessed with Bloqqer~\cite{DBLP:conf/cade/BiereLS11}). 
We also distinguish between true and false formulas. 
In Figure~\ref{fig:scatterplot:1800:size} we see that for true
formulas, set $S$ tends to be larger than $A$, while for false instances 
the picture is less clear. 
Figure~\ref{fig:scatterplot:1800:time} shows the overall time needed for solving \(\psi_\forall\) (y-axis) and 
\(\psi_\exists\) (x-axis). In almost all cases, the solver that 
handles $\psi_\forall$ needs more time than the solver that handles 
$\psi_\exists$.  
This may be founded on the observation that 
many QBFs have considerably more existential variables 
than universal variables \cite{DBLP:journals/corr/LonsingECP18appendix},
hence the instantiations 
added to $\psi_\forall$ are much larger than the instantiations 
added to $\psi_\exists$.   

In Line~\ref{alg:basic:init} of Figure~\ref{alg:basic}, 
the set of universal assignments \(A\) 
is initialised with \emph{one} arbitrary assignment \(\alpha_0\). 
Obviously, the set  \(A\) may also be initialized with 
multiple assignments. In our current implementation,  
we initialize \(A\) with the assignments 
that set the variables of one universal quantifier block to 
\(\bot\) and the variables of all other universal quantifier blocks 
to \(\top\). The impact of various initialization heuristics
remains to be investigated in future work.

\begin{figure*}[t]
\begin{minipage}[t]{0.33\textwidth} 
\begin{center}
\small
    {\setlength\tabcolsep{0.15cm}
\begin{tabular}{lrrcc}
  \hline
\emph{Solver} & \multicolumn{1}{c}{\emph{S}} & \multicolumn{1}{c}{\emph{$\bot$}} & \emph{$\top$} & \emph{Time} \\
\hline
\revqfun & 220 & 145 & 75 & 572K \\ 
\ghostqcegar & 194 & 120 & 74 & 617K \\ 
\caqe & 170 & 128 & 42 & 656K \\ 
\rareqs & 167 & 133 & 34 & 660K \\ 
\depqbfprefixopt & 167 & 121 & 46 & 666K \\ 
\heretic & 163 & 133 & 30 & 664K \\ 
\qstsdefnobreaksym & 152 & 116 & 36 & 687K \\ 
\ijtihad & 150 & 127 & 23 & 684K \\ 
\quterandom & 130 & 91 & 39 & 720K \\ 
\qesto & 109 & 86 & 23 & 761K \\ 
\dynqbf & 72 & 38 & 34 & 826K \\ 
\hline
\end{tabular}
    }
\captionof{table}{Original instances.}
\label{tab:2017:original}
\end{center}
\end{minipage}
\begin{minipage}[t]{0.33\textwidth} 
\begin{center}
\small
    {\setlength\tabcolsep{0.15cm}
\begin{tabular}{lcrcc}
\hline
\emph{Solver} & \emph{S} & \multicolumn{1}{c}{\emph{$\bot$}} & \emph{$\top$} & \emph{Time} \\
\hline
\rareqs & 256 & 180 & 76 & 508K \\ 
\caqe & 251 & 168 & 83 & 522K \\ 
\heretic & 245 & 172 & 73 & 522K \\
\revqfun & 219 & 148 & 71 & 568K \\ 
\ijtihad & 217 & 156 & 61 & 564K \\ 
\qstsdefnobreaksym & 208 & 151 & 57 & 585K \\ 
\qesto & 196 & 137 & 59 & 610K \\ 
\depqbfprefixopt & 183 & 117 & 66 & 633K \\ 
\ghostqcegar & 163 & 100 & 63 & 670K \\ 
\quterandom & 154 & 109 & 45 & 682K \\ 
\dynqbf & 151 & 95 & 56 & 684K \\ 
\hline
\end{tabular}
    }
\captionof{table}{Preprocessing~by \bloqqer.}
\label{tab:2017:bloqqer}
\end{center}
\end{minipage}
\begin{minipage}[t]{0.33\textwidth} 
\begin{center}
\small
    {\setlength\tabcolsep{0.15cm}
\begin{tabular}{lrrrc}
  \hline
\emph{Solver} & \multicolumn{1}{c}{\emph{S}} & \multicolumn{1}{c}{\emph{$\bot$}} & \emph{$\top$} & \emph{Time} \\
\hline
\heretic & 92 & 81 & 11 & 195K \\ 
\depqbfprefixopt & 89 & 74 & 15 & 205K \\ 
\caqe & 88 & 73 & 15 & 204K \\ 
\rareqs & 82 & 78 & 4 & 211K \\ 
\qstsdefnobreaksym & 81 & 69 & 12 & 216K \\ 
\ijtihad & 80 & 73 & 7 & 212K \\ 
\revqfun & 78 & 75 & 3 & 224K \\ 
\quterandom & 70 & 60 & 10 & 236K \\ 
\qesto & 51 & 44 & 7 & 269K \\ 
\ghostqcegar & 45 & 39 & 6 & 279K \\ 
\dynqbf & 15 & 13 & 2 & 332K \\ 
\hline
\end{tabular}
    }
\captionof{table}{197 original instances with four or more quantifier blocks.}
\label{tab:2017:original:many:qblocks}
\end{center}
\end{minipage}
\end{figure*}
In Line \ref{alg:basic:upd-s} and Line \ref{alg:basic:upd-a}
our algorithm increases the size of \(S\) and \(A\) in each iteration 
of the main loop, as argued in Theorem~\ref{thm:termination}. In the worst case, this 
leads to an exponential increase in space consumption. 
Although we detect shared clauses among the instantiations, 
that alone is not enough to significantly reduce the space consumption. However, 
some of the  assignments found in an earlier iteration 
could become obsolete after better assignments were found. It is therefore beneficial to empty 
either \(S\) or \(A\) and then reconstruct them from \(\psi_\forall\) and \(\psi_\exists\), similarly to what is done in Line \ref{alg:basic:upd-s} and Line \ref{alg:basic:upd-a}. We evaluated 
several heuristics for scheduling these set resets, and we found that
resetting periodically and close to the memory limit works best.
The regular resetting of one set  
has a similar effect as restarts in SAT solvers, and 
we observed a considerable improvement in performance, 
especially in terms of memory consumption. 
Our implementation periodically resets the set \(A\), 
since experiments indicate that 
the resulting formula $\psi_\forall$ is much harder to solve
than $\psi_\exists$ as seen in Figure~\ref{fig:scatterplot:1800:time}.
Besides the aforementioned imbalance between universal and 
existential variables, it is also likely due to the structure of 
\(\psi_\exists\) which is a conjunction of formulas in disjunctive 
normal form. 
Note that this reset of $A$ does not affect the termination argument presented 
in Theorem~\ref{thm:termination}, since the sets $A$ and $S$ still complete each other. 

Finally, we extended the presented approach with orthogonal 
reasoning techniques like QCDCL~\cite{DBLP:series/faia/GiunchigliaMN09}
for exploiting the different strengths of $\forall$Exp+Res and 
Q-resolution, yielding a hybrid solver that smoothly integrates both 
solving paradigms. 
To this end, we implemented the  prototypical solver 
called  \heretic which pursues the following idea:   
The main loop of the algorithm shown in Figure~\ref{alg:basic} 
(Lines~\ref{alg:basic:refloop-begin}-\ref{alg:basic:refloop-end}) 
is extended in a sequential portfolio style such that
a QCDCL solver is periodically called. 
After each call, all clauses that were learned through QCDCL are 
added to \(\Pi.\Phi\), making them available in further iterations. 
These new clauses potentially exclude assignments that 
would otherwise be possible and that could result in more 
iterations of the main loop. 

The solver \heretic extends \ijtihad by additional invocations of 
the QCDCL solver \depqbfbat~\cite{DBLP:conf/cade/LonsingE17}. 
About every $30$ seconds, \depqbfbat is called and run for about $30$ 
seconds. The learnt clauses are obtained via the API of \depqbfbat. 
Leveraging learned cubes is subject to future work.

\section{Evaluation}
\label{sec:eval}


We evaluate non-recursive expansion as implemented
in our solvers \ijtihad and its hybrid variant \heretic on the benchmarks from the PCNF
track of the QBFEVAL'17 competition. All experiments 
were carried out on a cluster of
Intel Xeon CPUs (E5-2650v4, 2.20 GHz) running Ubuntu 16.04.1 with a CPU 
time limit of 1800 seconds and a memory limit of 7 GB.
We considered the following top-performing solvers from
QBFEVAL'17: 
\quterandom~\cite{DBLP:conf/sat/PeitlSS17},
\revqfun~\cite{DBLP:conf/aaai/Janota18},
\rareqs~\cite{Janota20161},
\caqe~\cite{DBLP:conf/fmcad/RabeT15,DBLP:conf/cav/Tentrup17},
\dynqbf~\cite{DBLP:conf/sat/CharwatW16},
\ghostqcegar~\cite{Janota20161,DBLP:conf/sat/KlieberSGC10},
\depqbfprefixopt~\cite{DBLP:conf/cade/LonsingE17},
\qesto~\cite{DBLP:conf/ijcai/JanotaM15}, and
\qstsdefnobreaksym~\cite{DBLP:conf/sat/0001JT16,DBLP:conf/aaai/0001JT16}. 
Our experiments are based on original benchmarks
without preprocessing and benchmarks preprocessed using
\bloqqer~\cite{DBLP:conf/cade/BiereLS11,DBLP:journals/jair/HeuleJLSB15} with a
timeout of two hours.\footnote{We refer to the appendix for additional
experiments.}
We included the 76 formulas already solved by \bloqqer in both benchmark sets.

Tables~\ref{tab:2017:original} and~\ref{tab:2017:bloqqer} show the total
numbers of solved instances (\emph{S}), solved unsatisfiable (\emph{$\bot$})
and satisfiable ones (\emph{$\top$}), and total CPU time including timeouts.  
In the following, we focus on a comparison of our solvers \ijtihad and \heretic with
\rareqs
(cf.~Figure~\ref{fig:scatter:2017:original:bloqqer:heretic:vs:rareqs}). Unlike
our solvers, \rareqs is based on a recursive implementation of
expansion.
\begin{figure}[h!]
\includegraphics[scale=1]{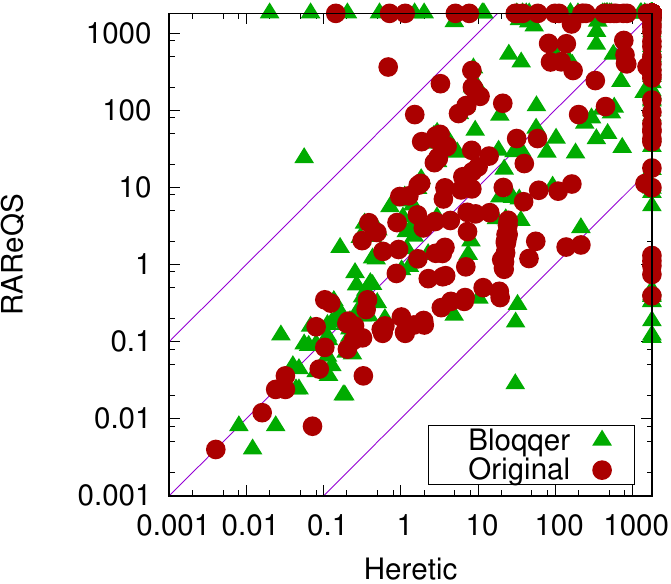}
\caption{Scatter plots of the run times of \heretic and \rareqs on original
  instances (related to Table~\ref{tab:2017:original}) and on instances
  preprocessed by \bloqqer (related to Table~\ref{tab:2017:bloqqer}).}
\label{fig:scatter:2017:original:bloqqer:heretic:vs:rareqs}
\end{figure}

In general, preprocessing has a considerable impact on the number of
solved instances. The difference in solved instances between \ijtihad
and \rareqs is 17 on original instances (Table~\ref{tab:2017:original}), and becomes larger  on
preprocessed instances (Table~\ref{tab:2017:bloqqer}).

Notably \heretic, despite its simple design, significantly
outperforms \ijtihad on the two benchmark sets. Moreover, \heretic
is ranked third on preprocessed instances 
(Table~\ref{tab:2017:bloqqer})
and thus is on par with state-of-the-art solvers. On the two
benchmark sets, the
gap in solved instances between \rareqs and \heretic is considerably
smaller than the one between \rareqs and \ijtihad.

We report on memory consumption of expansion-based solvers.  
While \rareqs, \ijtihad, and \heretic run out of memory on 42, 61, and 39
original instances (Table~\ref{tab:2017:original}), respectively, these numbers drop to
17, 41, and 24, respectively, with preprocessing 
(Table~\ref{tab:2017:bloqqer}). The average memory footprint is 1718~MB, 1836~MB, and 1842~MB for \rareqs, \ijtihad, and
\heretic, respectively, and 1056~MB, 1311~MB, and~1187~MB on preprocessed
instances. Interestingly, \ijtihad has a smaller median memory footprint than
\rareqs without (792~MB vs.~802~MB) and with preprocessing (286~MB vs.~364~MB).

The strength of \heretic becomes obvious for
formulas that have four or more quantifier blocks
(i.e., three or more quantifier alternations), cf.~\cite{DBLP:journals/corr/LonsingECP18appendix}. As
shown in Table~\ref{tab:2017:original:many:qblocks}, \heretic
outperforms all other solvers on these instances. We made a
similar observation on preprocessed formulas.

Moreover, \heretic solves only four original instances less than
\depqbfprefixopt (Table~\ref{tab:2017:original}), and outperforms
\depqbfprefixopt on preprocessed instances
(Table~\ref{tab:2017:bloqqer}). These results indicate the potential of
combining the orthogonal proof systems $\forall$Exp+Res as implemented in \ijtihad and
Q-resolution as implemented in \depqbfprefixopt in a hybrid solver such as
\heretic.

\begin{table}[h]
  {\setlength\tabcolsep{0.125cm}
    \begin{center}
  \begin{tabular}{l@{\qquad}ccc@{\qquad}ccc@{\qquad}ccc@{\qquad}ccc}
    \hline
        & \emph{R} & \emph{vs.} & \emph{I} & \emph{R} & \emph{vs.} & \emph{H} & \emph{I} & \emph{vs.} & \emph{H} &  \emph{D} & \emph{vs.} & \emph{H} \\
               & $<$ & = & $>$ & $<$ & = & $>$ & $<$ & = & $>$ & $<$ & = & $>$ \\ 
    N          & 27 & 140 & 10 & 26 & 141 & 22 & 5 & 145 & 18 & 65 & 102 & 61 \\
    B & 56 & 200 & 17 & 38 & 218 & 27 & 7 & 210 & 35 & 17 & 166 & 79 \\
    \hline
  \end{tabular}
    \end{center}
  } 
  \caption{Pairwise comparison of \rareqs (R), \ijtihad (I), 
    \heretic (H), and \depqbfprefixopt (D) by instances without (N) and with preprocessing by $\bloqqer$ (B) that were solved by only one solver of the considered pair ($<$,
    $>$) or by both ($=$).}
  \label{tab:orthogonal}
\end{table}
Although \rareqs outperforms both \ijtihad and \heretic on the two given  
benchmark sets (Tables~\ref{tab:2017:original} and~\ref{tab:2017:bloqqer}), \rareqs failed to solve certain instances that were
solved by \ijtihad and \heretic. Table~\ref{tab:orthogonal} shows related statistics. E.g., on
preprocessed instances (row ``B''), 218 instances were solved by both
\rareqs and \heretic (column ``\emph{R vs.~H}), 38 only by \rareqs,
and 27 only by \heretic. Summing up these numbers yields a total of
283 solved instances (more than any individual solver on preprocessed
instances in
Table~\ref{tab:2017:bloqqer}) that could have been solved by a
\emph{hypothetical solver} combining \rareqs and \heretic. This
observation underlines the strength of expansion in
general and, in particular, of the hybrid approach implemented in
\heretic. \heretic solved a significant amount of instances not solved by
\rareqs, it clearly outperformed \ijtihad on all benchmarks (column ``\emph{I
  vs.~H}'') and \depqbfprefixopt on preprocessed ones (``\emph{D
  vs.~H}'').


\section{Conclusion}
\label{sec:concl}

We presented a novel non-recursive algorithm for expansion-based QBF solving that 
uses only two SAT solvers for incrementally refining the propositional 
abstraction and the negated propositional abstraction of a QBF. 
We gave a concise proof of termination and soundness and demonstrated
with several experiments that our prototype compares well with 
the state of the art. 
In addition to non-recursive expansion,
we also studied the impact of combining Q-resolution and
$\forall$Exp+Res in a hybrid approach. To this end, we coupled a QCDCL
solver and non-recursive expansion to make clauses derived by the
QCDCL solver available to the expansion solver.
Experimental results indicated that
the hybrid approach significantly outperforms our implementation of
non-recursive expansion indicating the potential of combining 
expansion-based approaches with Q-resolution which gives rise to an 
exciting direction of future work. Further, our current 
implementation supports only formulas in conjunctive normal form while 
in theory, our 
approach does not make any assumptions on the structure of the 
propositional part of the QBF. We also plan to investigate 
how this formula structure can be exploited for efficiently 
processing the negation of the formula.


\clearpage

\appendices

\section{Additional Experimental Data}

\subsection*{Preprocessing with \hqspre}

\begin{figure}[h]
\begin{minipage}{0.33\textwidth} 
\small
\begin{center}
    {\setlength\tabcolsep{0.15cm}
\begin{tabular}{lccrc}
\hline
\emph{Solver} & \emph{S} & \emph{$\bot$} & \multicolumn{1}{c}{\emph{$\top$}} & \emph{Time} \\
\hline
\caqe & 306 & 197 & 109 & 415K \\ 
\rareqs & 294 & 194 & 100 & 429K \\ 
\qesto & 287 & 194 & 93 & 443K \\ 
\revqfun & 281 & 190 & 91 & 453K \\ 
\heretic & 279 & 188 & 91 & 460K \\ 
\qstsdefnobreaksym & 264 & 179 & 85 & 484K \\ 
\depqbfprefixopt & 263 & 176 & 87 & 490K \\ 
\quterandom & 255 & 171 & 84 & 497K \\ 
\ijtihad & 250 & 176 & 74 & 500K \\ 
\ghostqcegar & 244 & 163 & 81 & 522K \\
\dynqbf & 233 & 156 & 77 & 528K \\ 
\hline
\end{tabular}
    }
\captionof{table}{Preprocessing~by \hqspre.}
\label{tab:2017:hqspre}
\end{center}
\includegraphics[scale=0.5]{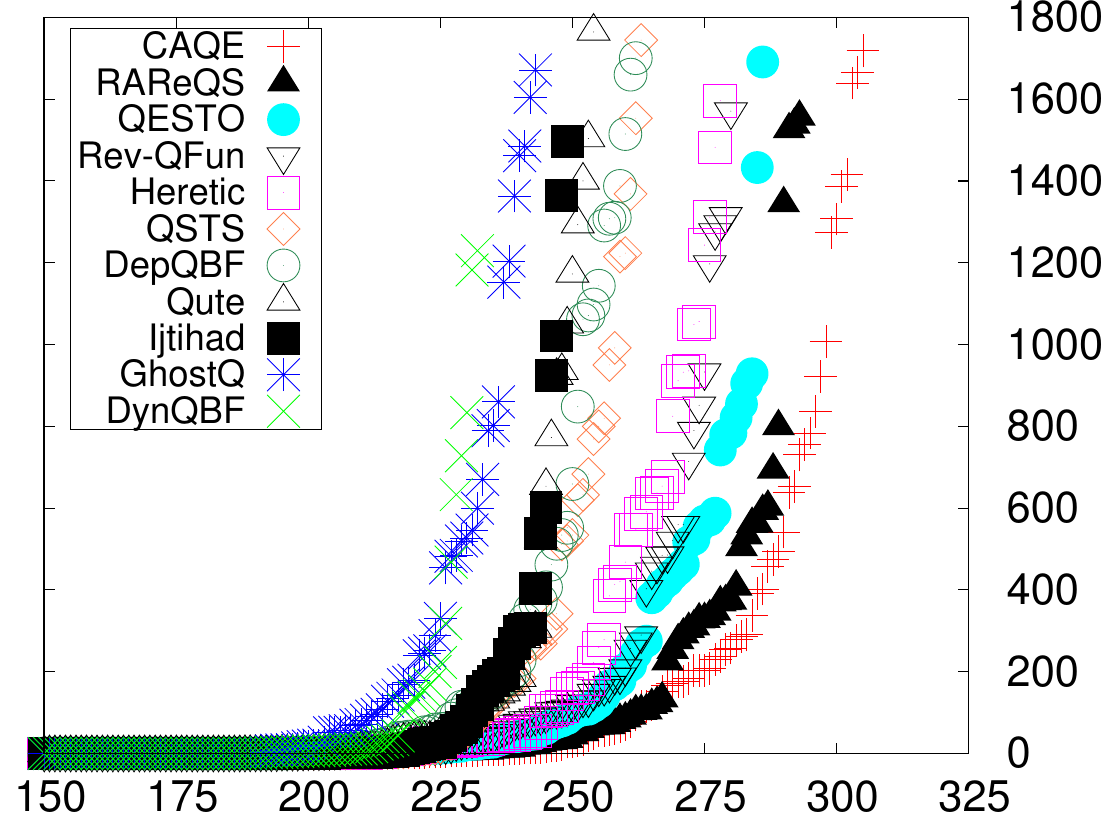}
\captionof{figure}{Plot related to Table~\ref{tab:2017:hqspre}.}
\label{fig:cactusplot:2017:hqspre}
\end{minipage}
\end{figure}

\begin{figure}[h]
\begin{minipage}{0.33\textwidth} 
\begin{center}
\small
    {\setlength\tabcolsep{0.15cm}
\begin{tabular}{lrrrc}
  \hline
\emph{Solver} & \multicolumn{1}{c}{\emph{S}} & \multicolumn{1}{c}{\emph{$\bot$}} & \emph{$\top$} & \emph{Time} \\
\hline
\depqbfprefixopt & 33 & 23 & 10 & 123K \\ 
\heretic & 30 & 20 & 10 & 127K \\ 
\qesto & 29 & 18 & 11 & 127K \\ 
\caqe & 29 & 16 & 13 & 127K \\ 
\quterandom & 24 & 17 & 7 & 137K \\ 
\rareqs & 23 & 17 & 6 & 136K \\ 
\revqfun & 22 & 18 & 4 & 137K \\ 
\qstsdefnobreaksym & 20 & 12 & 8 & 140K \\ 
\ijtihad & 17 & 11 & 6 & 145K \\ 
\ghostqcegar & 11 & 4 & 7 & 160K \\ 
\dynqbf & 7 & 4 & 3 & 163K \\ 
\hline
\end{tabular}
    }
\captionof{table}{Preprocessing by \hqspre: \\
97 instances with four or more \\quantifier blocks.}
\label{tab:2017:hqspre:many:qblocks}
\end{center}
\includegraphics[scale=0.5]{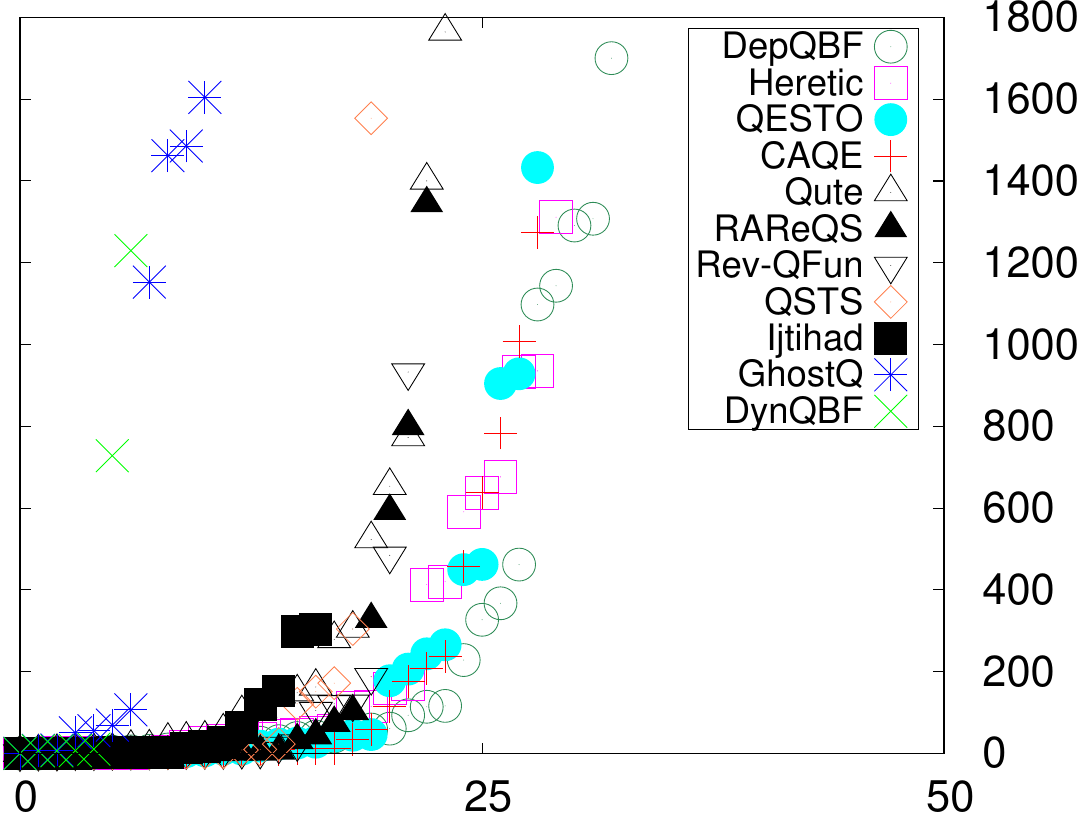}

\captionof{figure}{Plot related to Table~\ref{tab:2017:hqspre:many:qblocks}.}
\label{fig:cactusplot:2017:hqspre:many:qblocks}
\end{minipage}
\end{figure}

\begin{table}[h]
  {\setlength\tabcolsep{0.15cm}
    \begin{center}
  \begin{tabular}{l@{\qquad}ccc@{\qquad}ccc@{\qquad}ccc}
    \hline
        & \emph{R} & \emph{vs.} & \emph{I} & \emph{R} & \emph{vs.} & \emph{H} & \emph{I} & \emph{vs.} & \emph{H} \\
               & $<$ & = & $>$ & $<$ & = & $>$ & $<$ & = & $>$ \\ 
    No preprocessing          & 27 & 140 & 10 & 26 & 141 & 22 & 5 & 145 & 18 \\
    \bloqqer & 56 & 200 & 17 & 38 & 218 & 27 & 7 & 210 & 35 \\
    \hqspre  & 54 & 240 & 10 & 40 & 254 & 25 & 1 & 249 & 30 \\
    \hline
  \end{tabular}
    \end{center}
  } 
  \caption{Pairwise comparison of \rareqs (R), \ijtihad (I), and
    \heretic (H) by numbers of instances from QBFEVAL'17 with and
    without preprocessing that were solved by only one solver of the considered pair ($<$,
    $>$) or by both ($=$).}
  \label{tab:orthogonal:appendix}
\end{table}

\subsection*{Preprocessing with \bloqqer}

\begin{minipage}[t]{0.33\textwidth} 
\begin{center}
\small
    {\setlength\tabcolsep{0.15cm}
\begin{tabular}{lrrrc}
  \hline
\emph{Solver} & \multicolumn{1}{c}{\emph{S}} & \multicolumn{1}{c}{\emph{$\bot$}} & \emph{$\top$} & \emph{Time} \\
\hline
\heretic & 86 & 71 & 15 & 130K \\ 
\rareqs & 79 & 72 & 7 & 143K \\ 
\caqe & 73 & 59 & 14 & 161K \\ 
\qstsdefnobreaksym & 70 & 59 & 11 & 160K \\ 
\ijtihad & 67 & 57 & 10 & 158K \\ 
\revqfun & 64 & 55 & 9 & 168K \\ 
\depqbfprefixopt & 61 & 44 & 17 & 178K \\ 
\qesto & 51 & 43 & 8 & 190K \\ 
\quterandom & 46 & 40 & 6 & 208K \\ 
\ghostqcegar & 26 & 16 & 10 & 231K \\ 
\dynqbf & 23 & 13 & 10 & 241K \\ 
\hline
\end{tabular}
    }
\captionof{table}{154 preprocessed instances with four or more quantifier blocks.}
\label{tab:2017:bloqqer:many:qblocks}
\end{center}
\end{minipage}
\begin{minipage}{0.33\textwidth} 
\begin{center}
\includegraphics[scale=0.5]{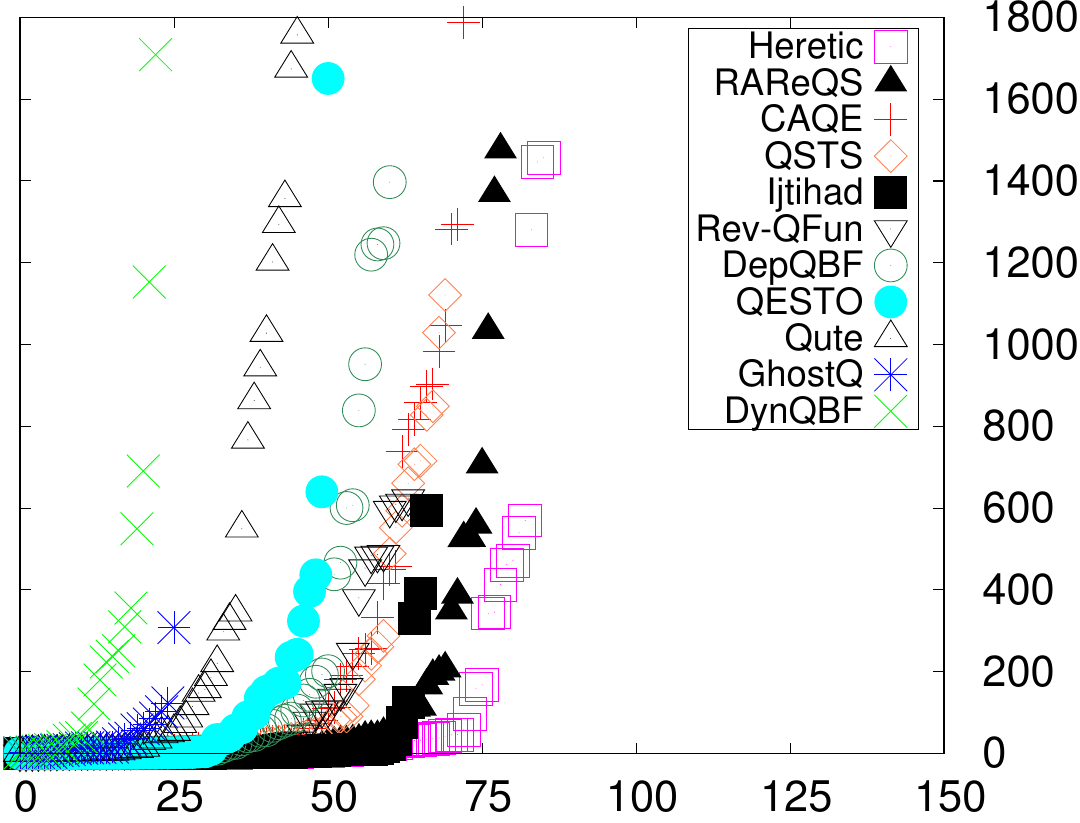}
\captionof{figure}{Plot related to Table~\ref{tab:2017:bloqqer:many:qblocks}.}
\label{fig:cactusplot:2017:bloqqer:many:qblocks}
\end{center}
\end{minipage}

\subsection*{Additional plots}

\begin{center}
\begin{minipage}[t]{0.33\textwidth} 
\includegraphics[scale=0.5]{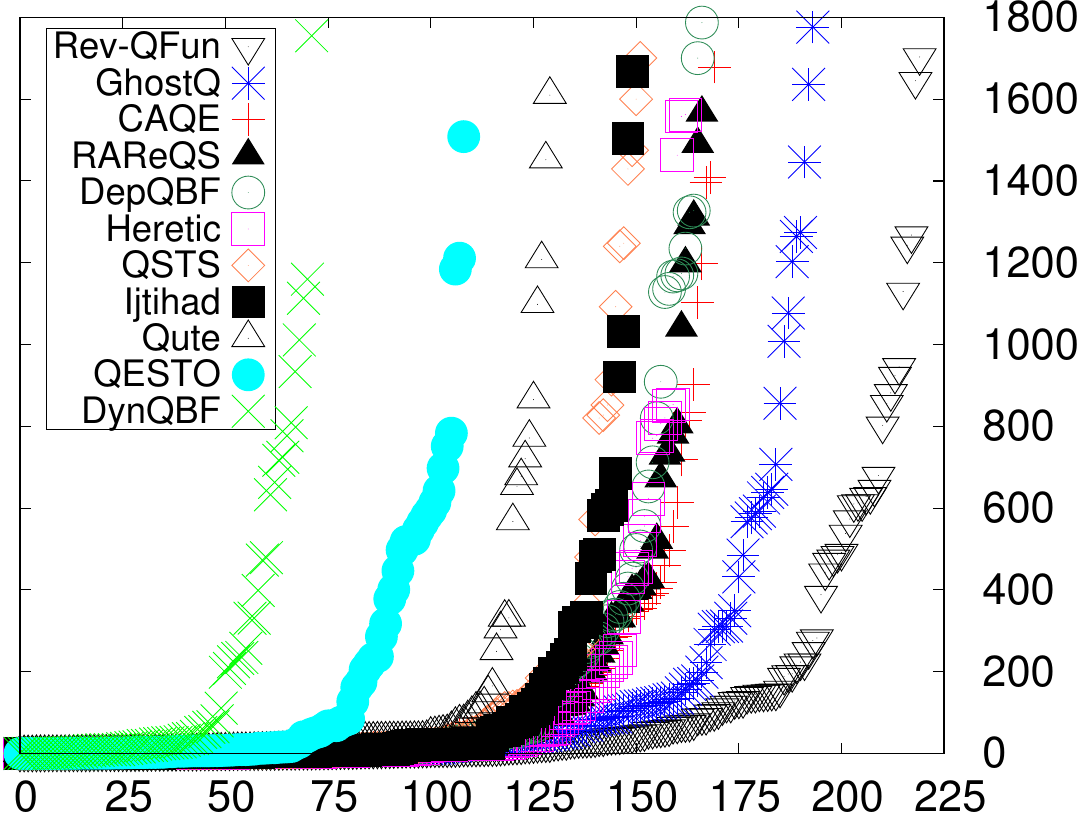}
\captionof{figure}{Plot related to Table~\ref{tab:2017:original}.}
\label{fig:cactusplot:2017:original}
\end{minipage}
\begin{minipage}[t]{0.33\textwidth} 
\includegraphics[scale=0.5]{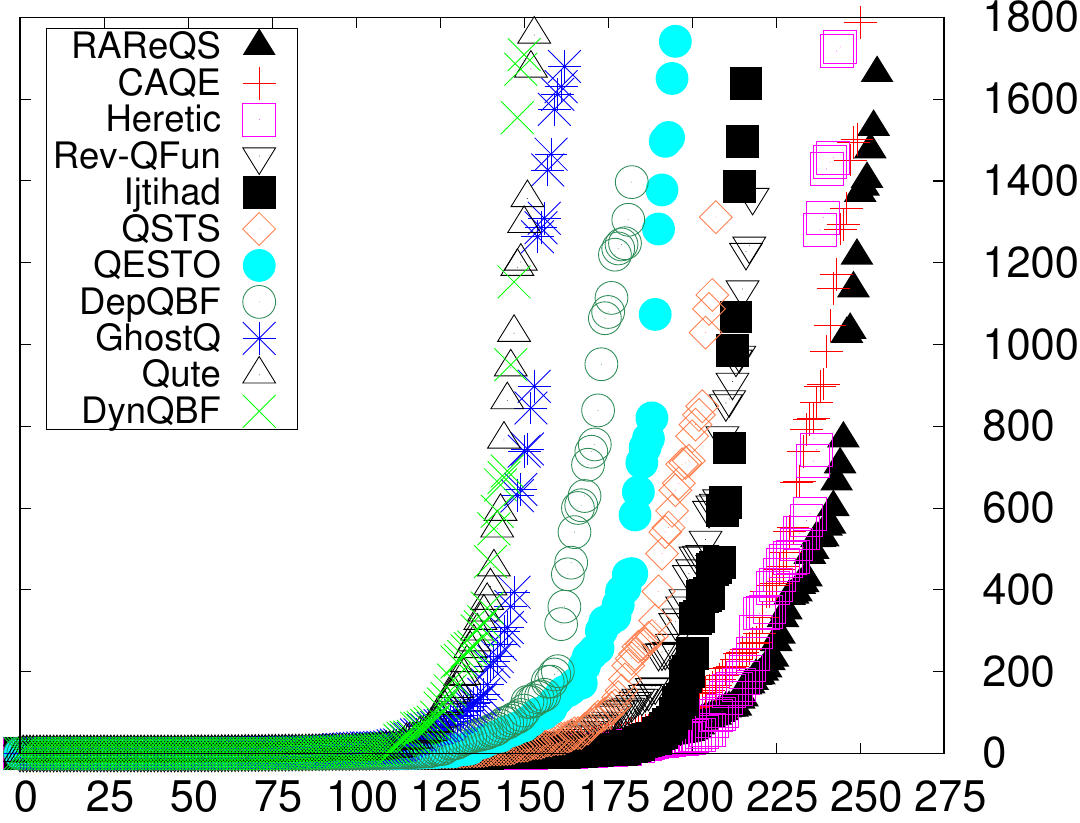}
\captionof{figure}{Plot related to Table~\ref{tab:2017:bloqqer}.}
\label{fig:cactusplot:2017:bloqqer}
\end{minipage}
\begin{minipage}{0.33\textwidth} 
\includegraphics[scale=0.5]{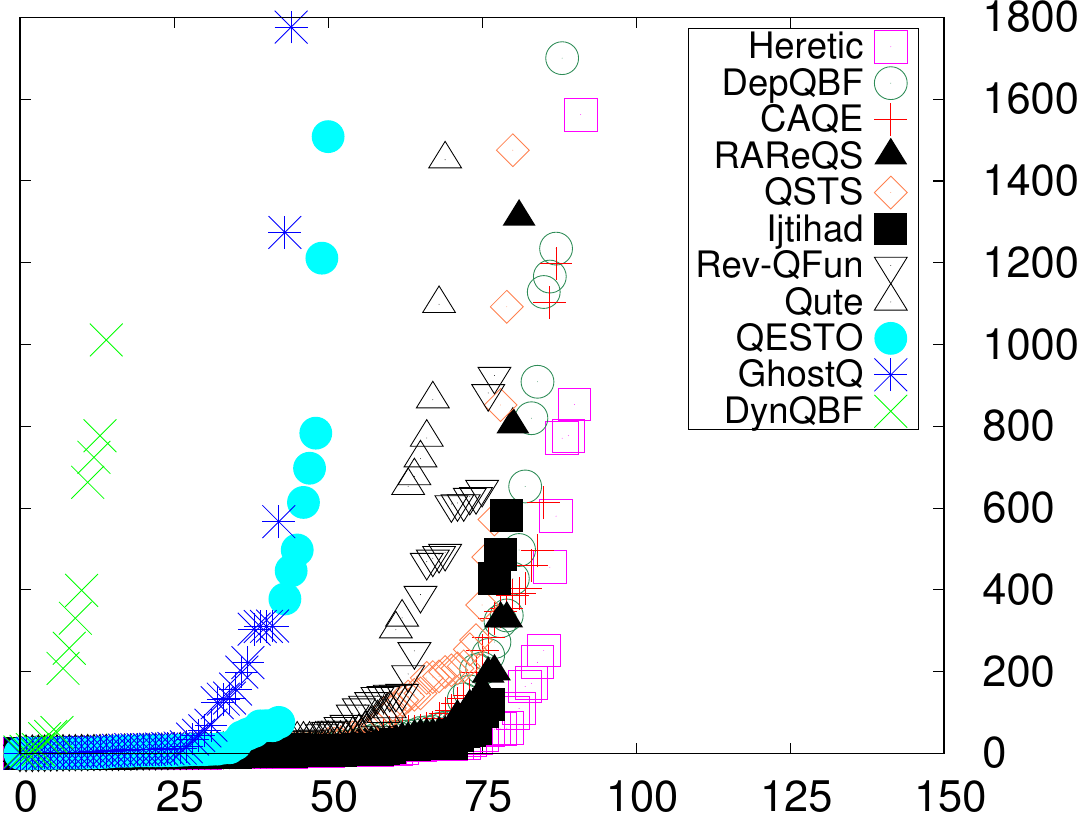}
\captionof{figure}{Plot related to Table~\ref{tab:2017:original:many:qblocks}.}
\label{fig:cactusplot:2017:original:many:qblocks}
\end{minipage}
\end{center}


\end{document}